\documentclass[11pt]{article} 

\title{An Innovations Algorithm for the prediction \\ of functional linear processes}

\author{Johannes Klepsch\thanks{Center for Mathematical Sciences, Technische Universit\"at M\"unchen,  85748 Garching, Boltzmannstrasse 3, Germany, e-mail: johannes.klepsch@tum.de\,,\,cklu@tum.de}
\and 
Claudia Kl\"uppelberg\footnotemark[1]
}

\usepackage[T1]{fontenc}
\usepackage[utf8]{inputenc}
\usepackage{amssymb}
\usepackage{amsfonts}
\usepackage{amsmath}
\usepackage{amsthm}
\usepackage{graphicx}
\usepackage{caption}
\usepackage{subcaption}
\usepackage[usenames, dvipsnames]{xcolor}
\usepackage{verbatim}
\usepackage{dsfont}
\usepackage{color}
\usepackage[numbers]{natbib}
\usepackage{relsize}
\usepackage{lmodern}
\usepackage{url}
\usepackage{siunitx}
\usepackage{MnSymbol}
\usepackage[english]{babel}
\usepackage{dsfont}  
\usepackage{tikz}
\usetikzlibrary{arrows}

\numberwithin{equation}{section}
\usepackage{caption} 
\newtheorem{theorem}{Theorem}[section]
\newtheorem{lemma}[theorem]{Lemma}
\newtheorem{remark}[theorem]{Remark}
\newtheorem{example}[theorem]{Example}
\newtheorem{proposition}[theorem]{Proposition}
\newtheorem{definition}[theorem]{Definition}

\newtheorem{corollary}[theorem]{Corollary}
\newtheorem{fig}[theorem]{Figure}

\newcommand{\bthe}{\begin{theorem}}
\newcommand{\ethe}{\end{theorem}}

\newcommand{\ben}{\begin{enumerate}}
\newcommand{\een}{\end{enumerate}}

\newcommand{\bit}{\begin{itemize}}
\newcommand{\eit}{\end{itemize}}

\newcommand{\beq}{\begin{equation}}
\newcommand{\eeq}{\end{equation}}

\newcommand{\ble}{\begin{lemma}}
\newcommand{\ele}{\end{lemma}}

\newcommand{\bde}{\begin{definition}\rm}
\newcommand{\ede}{\end{definition}}

\newcommand{\bco}{\begin{corollary}}
\newcommand{\eco}{\end{corollary}}

\newcommand{\bpr}{\begin{proposition}}
\newcommand{\epr}{\end{proposition}}

\newcommand{\brem}{\begin{remark}\rm}
\newcommand{\erem}{\end{remark}}

\newcommand{\bproof}{\begin{proof}}
\newcommand{\eproof}{\end{proof}}

\newcommand{\bexam}{\begin{example}\rm}
\newcommand{\eexam}{\end{example}}

\newcommand{\bfi}{\begin{fig}}
\newcommand{\efi}{\end{fig}}

\newcommand{\btab}{\begin{tab}}
\newcommand{\etab}{\end{tab}}

\newcommand{\beao}{\begin{eqnarray*}}
\newcommand{\eeao}{\end{eqnarray*}\noindent}

\newcommand{\beam}{\begin{eqnarray}}
\newcommand{\eeam}{\end{eqnarray}\noindent}

\newcommand{\barr}{\begin{array}}
\newcommand{\earr}{\end{array}}

\newcommand{\bdis}{\begin{displaymath}}
\newcommand{\edis}{\end{displaymath}\noindent}

\def\N{{\mathbb N}}
\def\C{{\mathbb C}}
\def\Z{{\mathbb Z}}
\def\P{{\mathbb P}}
\def\E{{\mathbb E}}
\def\R{{\mathbb R}}

\def\P{\mathbb{P}}

\def\cals_+{{\cals_+}}
\def\calb{{\mathcal{B}}}

\def\calf{{\mathcal{F}}}
\def\cala{{\mathcal{A}}}

\def\call{{\mathcal{L}}}
\def\cals{{\mathcal{S}}}

\def\calp{{\mathcal{P}}}

\def\caln{{\mathcal{N}}}
\def\calm{{\mathcal{M}}}

\newcommand{\al}{{\alpha}}

\newcommand{\la}{{\lambda}}

\newcommand{\eps}{\varepsilon}

\newcommand{\LCS}{{\rm LCS}}

\newcommand{\FMA}{{\rm FMA}}
\newcommand{\FAR}{{\rm FAR}}

\newcommand{\spa}{{\rm \ov{sp}}}

\newcommand{\ov}{\overline}
\newcommand{\wh}{\widehat}
\newcommand{\wt}{\widetilde}

\newcommand{\halmos}{\quad\hfill\mbox{$\Box$}}  

\newcommand{\JK}[1]{{\color{red} #1}}

\usepackage{paralist}

\usepackage[left=2.75cm,top=2cm,right=2.75cm,bottom=2cm,nohead,foot=1cm]{geometry}

\begin{document}


\maketitle

\begin{abstract}
When observations are curves over some natural time interval, the field of functional data analysis comes into play. 
Functional linear processes account for temporal dependence in the data. 
The prediction problem for functional linear processes has been solved theoretically, but the focus for applications has been on functional autoregressive processes.
We propose a new computationally tractable linear predictor for functional linear processes. 
It is based on an application of the Multivariate Innovations Algorithm to finite-dimensional subprocesses of increasing dimension of the infinite-dimensional functional linear process.
We investigate the behavior of the predictor for increasing sample size.
We show that, depending on the decay rate of the eigenvalues of the covariance and the spectral density operator, 
the resulting predictor converges with a certain rate to the theoretically best linear predictor.
\end{abstract}

\noindent
{\em AMS 2010 Subject Classifications:}  primary:\,\,\,62M10, 62M15, 62M20\,\,\,
secondary: \,\,\,62H25, 60G25


\noindent
{\em Keywords:}
functional linear process,
functional principal components, 
functional time series, 
Hilbert space valued process, 
Innovations Algorithm, 
prediction,
prediction error

\section{Introduction}

We consider observations which are consecutive curves over a fixed  time interval within the field of functional data analysis (FDA).
In this paper curves are representations of a functional linear process. 
The data generating process is a time series $X=(X_n)_{n\in\mathbb{Z}}$ where each $X_n$ is a random element $X_n(t)$, $t\in [0,1]$, of a Hilbert space, often the space of square integrable functions on $[0,1]$. 

Several books contain a mathematical or statistical treatment of dependent functional data as e.g. Bosq \cite{bosq}, Horv\`ath and Kokoszka \cite{horvath}, and Bosq and Blanke \cite{blanke}.
The main source of our paper is the book \cite{bosq} on linear processes in function spaces, which gives the most general mathematical treatment of linear dependence in functional data, developing estimation, limit theorems and prediction for functional  autoregressive processes.
In \cite{weaklydep} the authors develop limit theorems for  the larger class of weakly dependent functional processes. 
More recently, \cite{hoermann} and \cite{panaretros} contribute to frequency domain methods of functional time series.

Solving the prediction equations in function spaces is problematic and research to-date has mainly considered first order autoregressive models.
Contributions to functional prediction go hand in hand with an estimation method for the autoregressive parameter operator. 
The book \cite{bosq} suggests a Yule-Walker type moment estimator, spline approximation is applied in \cite{bessecardot}, and \cite{kargin} proposes a predictive factor method where the principal components are replaced by directions which may be more relevant for prediction.  

When moving away from the autoregressive process, results on prediction of functional time series become sparse. 
An interesting theory for the prediction of general functional linear processes is developed in \cite{bosq2014}. 
Necessary and sufficient conditions are derived for the best linear  predictor to take the form  $\phi_n(X_1,\dots,X_n)$ with $\phi_n$ linear and bounded. 
However, due to the infinite dimensionality of function spaces boundedness of $\phi_n$ cannot be guaranteed. 
Consequently, most results, though interesting from a theoretical point of view, are not suitable for application. 

More practical results are given for example in \cite{antoniadis}, where  prediction is performed non-parametrically with a functional kernel regression technique,  or in  \cite{aue}, \cite{hyndmanshang} and \cite{KKW}, where the dimensionality of the prediction problem is reduced via functional principal component analysis. 
In a multivariate setting, the Innovations Algorithm proposed in \cite{brockwell} gives a robust prediction method for linear processes. 
However, as often in functional data analysis, the non-invertibility of covariance operators prevents an ad-hoc generalization of the Innovations Algorithm to functional linear processes. 

We suggest a computationally feasible linear prediction method extending the Innovations Algorithm to the functional setting. 
For a functional linear process $(X_n)_{n\in\Z}$ with values in a Hilbert space $H$ and with innovation process $(\varepsilon_n)_{n\in\Z}$ our goal is a linear predictor $\wh{X}_{n+1}$ based on $X_1,\dots,X_n$ such that $\wh{X}_{n+1}$ is both computationally tractable and {\em consistent}.
In other words, we want  to find a bounded linear mapping $\phi_n$ 
with $\wh X_{n+1}=\phi_n(X_1,\dots,X_n)$ such that the statistical prediction error converges to $0$ for increasing sample size; i.e.,
\begin{align}\label{goal}
\lim_{n\to\infty} \E\Vert X_{n+1}-\wh{X}_{n+1}\Vert^2 = \E\Vert\varepsilon_0\Vert^2.
\end{align}
To achieve convergence in \eqref{goal} we work with finite dimensional projections of the functional process, similarly as in \cite{aue} and \cite{KKW}. 
We start with a representation of the functional linear model in terms of an arbitrary orthonormal basis of the Hilbert space. We then focus on a representation of the model based on only finitely many basis functions. 
An intuitive choice for the orthonormal basis consists of the eigenfunctions of the covariance operator of the process. 
Taking the eigenfunctions corresponding to the $D$ largest eigenvalues results in a truncated Karhunen-Lo\'eve representation, and guarantees to capture most of the variance of the process (see \cite{aue}). 
Other applications may call for a different choice.

Though the idea of finite dimensional projections is not new, our approach differs significantly from existing ones. 
Previous approaches consider the innovations of the projected process as the projection of the innovation of the original functional process. 
Though this may be sufficient in practice, it is in general not theoretically accurate. 

The Wold decomposition enables us to work with the exact dynamics of the projected process, which then allows us to derive precise asymptotic results. 
The task set for this paper is of a purely predictive nature: we assume knowing the dependence structure and do not perform model selection or covariance estimation.  
{This will be the topic of a subsequent paper. }

 The truncated process $(X_{D,n})_{n\in\Z}$ based on $D$ basis functions is called subprocess.  We show that every subprocess of a stationary (and invertible) functional process is again stationary (and invertible). We then use an isometric isomorphy to a $D$-dimensional vector process to compute the best linear predictor of $(X_{D,n})_{n\in\Z}$ by the Multivariate Innovations Algorithm (see e.g. \cite{brockwell}). 
 
As a special example we investigate the functional moving average process of finite order. We prove that every subprocess is again a functional moving average process of same order or less. Moreover, for this process the Innovations Algorithm simplifies. 
Invertibility is a natural assumption in the context of prediction (cf. \cite{brockwell}, Section 5.5, and \cite{nsiri}), and we require it when proving limit results.
The theoretical results on the structure of $(X_{D,n})_{n\in\Z}$ enable us to quantify the prediction error in \eqref{goal}. As expected, it can be decomposed in two terms, one due to the dimension reduction, and the other due to the statistical prediction error of the $D$-dimensional model. However, the goal of consistency as in \eqref{goal} is not satisfied, as the error due to dimension reduction does not depend on the sample size.

Finally, in order to satisfy \eqref{goal}, we propose a modified version of the Innovations Algorithm. The idea is to increase $D$ together with the sample size. Hence the iterations of our modified Innovations Algorithm are based on increasing subspaces. Here we focus on the eigenfunctions of the covariance operator of $X$ as orthonormal basis of the function space. 

Our main result states that the prediction error is a combination of two tail sums, one involving  operators of the inverse representation of the process, and the other the eigenvalues of the covariance operator. 
We obtain a computationally tractable functional linear predictor for stationary invertible functional linear processes. As the sample size tends to infinity the predictor satisfies \eqref{goal} with a rate depending on the eigenvalues of the covariance operator and of the spectral density operator.

Our paper is organized as follows. 
After  summarizing prerequisites of functional time series in Section~\ref{sec2}, we recall in Section~\ref{sec3} the framework of prediction in infinite dimensional Hilbert spaces, mostly based on the work of Bosq (see \cite{bosq,bosq2,bosq2014}). 
Here we also clarify the difficulties of linear prediction in infinite dimensional function spaces. 
In Section~\ref{sec4} we propose an Innovations Algorithm based on a finite dimensional subprocess of $X$.
The predictor proposed in Section~\ref{sec4}, though quite general,  does not satisfy \eqref{goal}.
Hence, in Section~\ref{sec5} we project the process on a finite-dimensional subspace spanned by the eigenfunctions of the covariance operator of $X$, and formulate the prediction problem in such a way that the dimension of the subprocess increases with the sample size. 
A modification of the Innovations Algorithm then yields a predictor, which satisfies \eqref{goal} and remains computationally tractable. 
The proof of this result requires some work and is deferred to Section~\ref{proofs} along with some auxiliary results.



\section{Methodology} \label{sec2}

Let $H =L^2([0,1])$ be the real Hilbert space of square integrable functions with norm $\|x\|_{H}=(\int_0^1 x^2(s)ds)^{1/2}$ generated by the inner product $\left\langle x , y \right\rangle=\int_0^1 x(s)y(s)ds$ for $x,y\in H$.  
We denote by $\mathcal{L}$ the space of bounded linear operators acting on $H$. 
If not stated differently, for $A\in\call$ we take the standard operator norm $\|A\|_{\mathcal{L}} = \sup_{ \|x\|\leq 1} \|Ax\|$.
Its adjoint $A^*$ is defined by $\langle Ax,y \rangle = \langle x , A^* y\rangle$ for $x,y\in H$.   
The operator $A\in\call$ is called { nuclear operator} (denoted by $\caln$), if 
$\Vert A\Vert_{\caln}=\sum_{j=1}^{\infty} \langle A e_j,e_j\rangle<\infty$ for some (and hence all) orthonormal basis (ONB) $(e_j)_{j\in\N}$ of $H$. 
If additionally $A$ is self-adjoint, then $\Vert A\Vert_{\caln}=\sum_{j=1}^\infty \vert \la_j\vert <\infty$, where $(\la_j)_{j\in\N}$ are the eigenvalues of $A$.
We shall also use the estimate $\|A B\|_{\caln} \le \|A\|_\call \|B\|_{\caln}$ for $A\in\call$ and $B\in\caln$.
For an introduction and more insight into Hilbert spaces we relied on Chapters~3.2 and 3.6 in \cite{simon}.

Let $\calb_H$ be the Borel $\sigma$-algebra of subsets of $H$.
All random functions are defined on a probability space $(\Omega,\mathcal{A},\mathcal{P})$ and are $\cala-\calb_H$-measurable.
Then the space of square integrable random functions  $L^2_{H}:=L^2(\Omega, \mathcal{A},\mathcal{P})$ is a Hilbert space with inner product 
$\mathbb{E}\left\langle X,Y\right\rangle=\E\int_0^1X(s)Y(s)ds$ for $X,Y \in L^2_{H}$. 
Furthermore, we say that $X$ is { integrable} if $\E \Vert X \Vert = \E [ (\int_0^1 X^2(t) dt) ^ {1 /2}] < \infty$. 

From Lemma~1.2 of \cite{bosq} we know that $X$ is a random function with values in $H$ if and only if $\langle \mu, X\rangle$ is a real random variable for every $\mu\in H$.
Hence, the following definitions are possible.

\bde\label{moments}
(i) \, If $X \in L^2_{H}$ is integrable, then there exists a unique $\mu \in H$ such that $\E \langle y, X \rangle = \langle y , \mu \rangle$ for $y \in H$.
It follows that $\E X(t) = \mu(t)$ for almost all $t \in [0,1]$, and $\E X \in H$ is called the {\em expectation} of $X$.\\
(ii) \,   If $X \in L^2_{H}$ and $\E X=0\in H$, the {\em covariance operator} of $X$ is defined as
\begin{align*}
	C_X(y) = \E [\langle X, y \rangle X], \quad y \in H.
\end{align*}
(iii) \,  If $X,Y \in L^2_{H}$ and $\E X = \E Y=0$, the {\em cross covariance operator} of $X$ and $Y$ is defined as 
\begin{align*}
	C_{X,Y}(y) = C_{Y,X}^*(y)=\E [\langle X, y \rangle Y],  \quad y \in H.
\end{align*}
\ede

The operators $C_X$ and $C_{Y,X}$ belong to $\caln$ (cf. \cite{bosq}, Section~1.5). 
Furthermore, $C_X$ is a self-adjoint ($C_X=C_X^*$) and non-negative definite operator with spectral representation
\begin{align*}
	C_X(x)=\sum_{j=1}^{\infty} \lambda_j \langle x, \nu_j \rangle \nu_j, \quad x\in H, 
\end{align*}
for eigenpairs $(\lambda_j,\nu_j)_{j\in \N}$, where $(\nu_j)_{j\in\N}$ is an ONB of $H$ and $(\lambda_j)_{j\in\N}$ is a sequence of positive real numbers such that $\sum_{j=1}^{\infty} \lambda_j < \infty$. 
When considering spectral representations, we assume that the $\lambda_j$ are ordered decreasingly; i.e., $\lambda_i \geq \lambda_k$ for $i<k$.

For ease of notation we introduce the operator 
\begin{align*}
x\otimes y(\cdot)= \langle x ,\cdot \rangle  y, 
\end{align*}
 which allows us to write $C_X = \E [  X\otimes X  ]$ and $C_{X,Y} = \E[X\otimes Y]$. 
 Note also that
 \begin{align}
 \E\Vert X\Vert^2 = \E\Vert X\otimes X\Vert_{\caln}=\Vert C_X\Vert_{\caln}. \label{traceeq}
 \end{align}
  Additionally, the following equalities are useful: for $A\in \call$  and $x_i,y_i\in H$ for $i=1,2$ we have 
\beam\label{otimesprop}
 \barr{rcl}
 A(x_1\otimes y_1) &=& A( \langle x_1, \cdot \rangle y_1 )= \langle x_1,\cdot \rangle Ay_1= x_1\otimes Ay_1,\\
 (x_1+x_2)\otimes(y_1+y_2) &=& x_1\otimes y_1+ x_1\otimes y_2 + x_2 \otimes y_1 + x_2 \otimes y_2.
 \earr
 \eeam
 
We define now functional linear processes and state some of their properties, taken from \cite{bosq},  Section~1.5 and Section~3.1.
We first define the driving noise sequence. 

\bde
 $(\varepsilon_n)_{n\in\mathbb{Z}}$ is {\em white noise (WN) in $L^2_{H}$} if $\E\, \varepsilon_n=0$, $0< \E\|\eps_n\|^2=\sigma^2<\infty$, $C_{\varepsilon_n}=C_{\varepsilon}$ is independent of $n$, and if $C_{\eps_n, \eps_{m}}=0$ for all $n,m\in\Z$, $n\neq m$. 
\ede

\bde[\cite{bosq}, Definition~7.1]
	Let $(\varepsilon_n)_{n\in\mathbb{Z}}$ be WN and $(\psi_j)_{j\in \mathbb{N}}$ a sequence in $\mathcal{L}$. Define $\psi_0=I_H$, the identity operator on $H$, and let $\mu\in H$. We call $(X_n)_{n\in\Z}$ satisfying
	\begin{align}\label{process}
		X_n =\mu + \sum_{j=0}^{\infty} \psi_j\varepsilon_{n-j}, \quad n\in\mathbb{Z}, 
	\end{align}
	a {\em functional linear process} in $L_H^2$ with mean $\mu$. The series in \eqref{process} converges in probability.
\ede

Note that by definition a functional linear process is causal.
We now state assumptions to ensure stronger convergence of the above series. 

\begin{lemma}[\cite{bosq}, Lemma~7.1(2)]\label{lemstationary}
 Let $(\varepsilon_n)_{n\in\mathbb{Z}}$ be WN and $\sum_{j=0}^{\infty} \Vert \psi_j \Vert_{\mathcal{L}}^2  < \infty $. 
	Then the series in \eqref{process} converges in $L^2_H$ and a.s., and $(X_n)_{n\in\Z}$ is (weakly) stationary.
\end{lemma}

Strict stationarity of a functional linear process can be enforced by assuming that  $(\varepsilon_n)_{n\in\Z}$ is additionally independent. 
In our setting weak stationarity will suffice. From here on, without loss of generality we set $\mu=0$. For a stationary process $(X_n)_{n\in\Z}$, the covariance operator with lag $h$ is denoted by
\begin{align}
C_{X;h}=\E[X_0\otimes X_h],\quad h\in \Z. \label{cxh}
\end{align}

We now define the concept of invertibility of a functional linear process, which is a natural assumption in the context of prediction; cf.  \cite{brockwell}, Chapter~5.5 and \cite{nsiri}. 

\bde[\cite{merlevede}, Definition~2]\label{definvertible}
A functional linear process $(X_n)_{n\in\Z}$ is said to be {\em invertible} if it admits the representation
	\begin{align}
		X_n= \varepsilon_n + \sum_{j=1}^{\infty} \pi_j X_{n-j}, \quad n\in\mathbb{Z}, 	\label{invertible}
	\end{align}
for $\pi_j\in \call$ and $\sum_{j=1}^\infty \Vert \pi_j \Vert_{\call} < \infty$.
\ede

Note that, as for univariate and multivariate time series models, every stationary causal  functional autoregressive moving average (FARMA) process is a functional linear process (see \cite{spangenberg}, Theorem~2.3). 
Special cases include functional autoregressive processes of order $p\in\N$ (\FAR$(p)$), which have been thoroughly investigated. 
Our focus is on functional linear models, with the {functional moving average process of order $q\in\N$} (\FMA$(q)$) as illustrating example, which we investigate in Section~\ref{4.2}.

\bde\label{defFMA}
	For $q\in\N$ a $\FMA(q)$ is a functional linear process $(X_n)_{n\in \mathbb{Z}}$ in $L^2_{H}$ such that
	\begin{align}
		X_n = \varepsilon_n + \sum_{j=1}^q \gamma_j \varepsilon_{n-j} , \quad n \in \mathbb{Z},  \label{FMA}
	\end{align}
	for WN $(\eps_n)_{n\in\Z}$  and $\gamma_j \in \call$ for $j=1,\dots,q$. 
\ede

A \FMA$(q)$ process can be characterized as follows:

\begin{proposition} [\cite{blanke}, Prop.~10.2] \label{maq}
	A  stationary  functional linear process $(X_n)_{n\in\mathbb{Z}}$ in $L^2_{H}$ is a \FMA$(q)$ for some $q\in\N$ if and only if $C_{X;q}\neq 0$  and $C_{X;h}=0$ for $|h|>q$.
\end{proposition}

\section{Prediction in Hilbert spaces} \label{sec3}

In a finite dimensional setting, when the random elements take values in  $\mathbb{R}^d$ equipped with the Euclidean norm, the concept of linear prediction of a 
random vector is well known (e.g. \cite{brockwell}, Section~11.4).
The best linear approximation of a random vector $\mathbf{X}$ based on vectors $(\mathbf{X}_1,\dots,\mathbf{X}_n$ is the orthogonal projection of each component of $\mathbf{X}$ on the smallest closed linear subspace of $L^2_{\R}(\Omega,\mathcal{A},\P)$ generated by the components of $\mathbf{X}_i$. 
This results in
$$\wh{\mathbf{X}} :=\sum_{i=1}^n \mathbf{\Phi}_{n,i}\mathbf{X}_i$$
for $\mathbf{\Phi}_{n,i}\in \R^{d\times d}$.
In infinite dimensional Hilbert spaces we proceed similarly, but need a rich enough subspace on which to project. The concept of linear prediction in infinite dimensional Hilbert spaces was introduced by Bosq; see Section~1.6 in \cite{bosq}. 
We start by recalling the notion of $\mathcal{L}$-closed subspaces (\LCS), introduced in \cite{fortet}.

\bde \label{lcs}
	$\mathcal{G}$ is said to be an {\em $\mathcal{L}$-closed subspace} {\rm (\LCS)} of $L^2_{H}$  if $\mathcal{G}$ is a Hilbertian subspace of $L^2_{H}$, and if $X\in \mathcal{G}$ and $l\in\mathcal{L}$ imply $lX \in \mathcal{G}$.
\ede

We now give a characterization of a \LCS\ generated by a subset of $L^2_{H}$.

\bpr[\cite{bosq}, Theorem 1.8] \label{theo18}
	Let $F \subseteq L^2_{H}$. 
	Then the \LCS\ generated by $F$, denoted by \LCS$(F)$, is the closure with respect to $\Vert \cdot \Vert$ of
	\begin{align*}
	F'=\Big\{\sum_{i=1}^k l_iX_i: \ l_i \in \mathcal{L}, \ X_i \in F, \ k\geq 1 \Big\}.
	\end{align*}
\epr

We are now ready to  define the best linear predictor in an infinite dimensional Hilbert space analogous to the finite dimensional setting. 

\bde\label{defblp}
Let 
 $X_1,\dots,X_n$ be zero mean random elements in $L^2_{H}$. 
Define
\begin{align}
F_n=\{X_1,\dots,X_n\}\quad\mbox{and}\quad \wh{X}_{n+1}=P_{\LCS(F_n)}( X_{n+1}), \label{blp}
\end{align}
i.e., $\wh{X}_{n+1}$ is the orthogonal projection of $X_{n+1}$ on $\LCS(F_n)$.
Then $\wh X_{n+1}$ is called {\em best linear functional predictor} of $X_{n+1}$ based on $\LCS(F_n)$.
\ede

Note however that, since $F'$ is not closed, $\wh{X}_{n+1}$ as in \eqref{blp} has in general not the form $\wh{X}_{n+1}=\sum_{i=1}^n l_iX_i$ for $l_i\in\call$ (e.g. \cite{bosq2014}, Proposition~2.2). 
For functional linear processes the above representation is purely theoretical. 
In the following we develop an alternative approach based on finite dimensional projections of the functional process.

\section{Prediction based on a finite dimensional projection} \label{sec4}

For a stationary functional linear process $(X_n)_{n\in\Z}$
the infinite dimensional setting makes the computation of $\wh{X}_{n+1}$ as in \eqref{blp} basically impossible. 
A natural solution lies in finite dimensional projections of the functional process $(X_n)_{n\in\Z}$. 
For fixed $D\in\N$ we define
\begin{align}\label{AD}
	(\nu_i)_{i\in\N}\quad\mbox{and}\quad A_{D}=\spa\{\nu_1,\dots,\nu_{D}\},
\end{align}
where $(\nu_i)_{i\in\N}$ is some ONB of $H$, and consider the projection of a functional random element on $A_D$.
In \cite{aue} and \cite{KKW} the authors consider the projection of  a \FAR\ process $(X_n)_{n\in\Z}$ on $A_D$, where $\nu_1,\dots,\nu_D$ are the eigenfunctions corresponding to the largest eigenvalues of $C_X$. 
However, instead of considering the true dynamics of the subprocess, they work with an approximation which lies in the same model class as the original functional process; e.g. projections of functional AR$(p)$ models are approximated by multivariate AR$(p)$ models. The following examples clarifies this concept.

\bexam \label{exar}
 Consider a \FAR$(1)$ process $(X_n)_{n\in\Z}$ as defined in Section~3.2 of \cite{bosq} by 
 \begin{align}
 X_n= \Phi X_{n-1} + \varepsilon_n, \qquad n\in\Z \label{FAR},
 \end{align}
 for some $\Phi\in\call$ and WN $(\varepsilon_n)_{n\in\Z}$. 
Let furthermore $(\nu_i)_{i\in\N}$ be an arbitrary ONB of $H$. 
Then \eqref{FAR} can be rewritten in terms of $(\nu_i)_{i\in\N}$ as
\begin{small}
\begin{align*}
\begin{pmatrix}
  \left\langle X_{n},\nu_1\right\rangle \\
  \vdots\\
  \left\langle X_{n},\nu_D\right\rangle  \\
 \hline
  \left\langle X_{n},\nu_{D+1}\right\rangle\\
  \vdots
 \end{pmatrix}&=\left(\begin{array}{ccc|cc}
\left\langle\phi\nu_1,\nu_1\right\rangle  &\dots &\left\langle\phi\nu_D,\nu_1\right\rangle &\left\langle\phi\nu_{D+1},\nu_1\right\rangle & \dots \\
\vdots & \ddots &\vdots &\vdots &\ddots\\
\left\langle\phi\nu_1,\nu_D\right\rangle  &\dots &\left\langle\phi\nu_D,\nu_D\right\rangle &\left\langle\phi\nu_{D+1},\nu_D\right\rangle &\dots\\
\hline
\left\langle\phi\nu_1,\nu_{D+1}\right\rangle  &\dots &\left\langle\phi\nu_D,\nu_{D+1}\right\rangle &\left\langle\phi\nu_{D+1},\nu_{D+1}\right\rangle &\dots\\
\vdots & \ddots &\vdots &\vdots &\ddots
\end{array}
\right)
\begin{pmatrix}
  \left\langle X_{n-1},\nu_1\right\rangle \\
  \vdots\\
  \left\langle X_{n-1},\nu_D\right\rangle  \\
  \hline
  \left\langle X_{n-1},\nu_{D+1}\right\rangle\\
  \vdots
 \end{pmatrix}
+\begin{pmatrix}
    \left\langle \varepsilon_{n},\nu_1\right\rangle \\
    \vdots\\
    \left\langle \varepsilon_{n},\nu_D\right\rangle  \\
    \hline
    \left\langle \varepsilon_{n},\nu_{D+1}\right\rangle\\
    \vdots
   \end{pmatrix},
\end{align*}
 \end{small}
 which we abbreviate as
\begin{equation}\label{ch4_bockrep}
\begin{pmatrix}
\mathbf{X}_{D,n}\\
\hline
\mathbf{X}_{n}^\infty
\end{pmatrix}
=\left[
\begin{array}{c|c}
\mathbf{\Phi}_D & \mathbf{\Phi}_D^\infty \\
\hline
\vdots\ &\vdots
\end{array}
\right]\begin{pmatrix}
\mathbf{X}_{D,n-1}\\
\hline
\mathbf{X}_{n-1}^\infty
\end{pmatrix}+\begin{pmatrix}
\mathbf{E}_{D,n}\\
\hline
\mathbf{E}_{n}^\infty
\end{pmatrix}.
\end{equation}
We are interested in the dynamics of the $D$-dimensional subprocess $(\mathbf{X}_{D,n})_{n\in\Z}$.  
From \eqref{ch4_bockrep} we find that 
 $(\mathbf{X}_{D,n})_{n\in\Z}$ satisfies
\begin{align}\label{ch4_vectorreal}
\mathbf{X}_{D,n}
&=\mathbf{\Phi}_D\mathbf{X}_{D,n-1}+\mathbf{\Phi}_D^{\infty}\mathbf{X}_{n-1}^{\infty}+\mathbf{E}_{D,n},\quad n\in\Z,
\end{align}
which does in general {not} define a \FAR(1) process. 
This can be seen from the following example, similar to Example~3.7 in \cite{bosq}. For some $a\in\R$ with $0<a<1$ let
 \begin{align*}
 \Phi(x)=a \sum_{j=1}^{\infty} \langle x,\nu_j\rangle \nu_1 + a \sum_{i=1}^{\infty}\langle x,\nu_i \rangle \nu_{i+1}, \qquad x\in H.
 \end{align*}
 Furthermore, assume that $\E\langle \varepsilon_n,\nu_1\rangle ^2 >0$ but $\E\langle \varepsilon_n,\nu_j\rangle ^2 =0$ for all $j>1$. 
 Since $\Vert \Phi \Vert_\call =a < 1$, $(X_n)_{n\in\Z}$ 
 is a stationary \FAR(1) process. 
 However, with \eqref{ch4_vectorreal} for $D=1$,
 \begin{align*}
 \mathbf{X}_{1,n}&=\langle X_n,\nu_1 \rangle = a \sum_{j=1}^{\infty} \langle X_{n-1},\nu_j \rangle +\langle \varepsilon_n,\nu_1\rangle\\
 &= a \langle X_{n-1},\nu_1\rangle + a\sum_{j=2}^{\infty} \Big\langle\big(a \sum_{j'=1}^{\infty} \langle X_{n-2},e_{j'}\rangle \nu_1 + a \sum_{i=1}^{\infty}\langle X_{n-2},\nu_i\rangle e_{i+1} + \varepsilon_{n-1}\big),\nu_j \Big\rangle + \langle \varepsilon_{n},\nu_1\rangle \\
 &=a \langle X_{n-1},\nu_1\rangle + a ^2 \langle X_{n-2},\nu_1\rangle + a ^2\sum_{j=2}^{\infty} \langle X_{n-2},\nu_j \rangle + \langle \varepsilon_{n},\nu_1 \rangle\\
 &=\sum_{j=1}^{\infty}a^j \mathbf{X}_{1,n-j} + \mathbf{E}_{n,1}.
 \end{align*}   
 Hence, $(\mathbf{X}_{1,n})_{n\in\Z}$ follows an AR$(\infty)$ model and $(\mathbf{X}_{1,n}\nu_1)_{n\in\Z}$ a \FAR$(\infty)$ model.
 
 In \cite{aue} and \cite{KKW},  $(\mathbf{X}_{D,n})_{n\in\Z}$ is approximated by  $(\mathbf{\tilde{X}}_{D,n})_{n\in\Z}$ satisfying
\begin{align*}
\mathbf{\tilde{X}}_{D,n}
&=\mathbf{\Phi}_D\mathbf{\tilde{X}}_{D,n-1}+\mathbf{E}_{D,n},\quad n\in\Z,
\end{align*}
such that  $(\mathbf{\tilde{X}}_{D,n})_{n\in\Z}$ follows a vector  AR(1) process.
\eexam

We pursue the idea of Example~\ref{exar} for functional linear processes and work with the true dynamics of a finite-dimensional subprocess.

\subsection{Prediction of functional linear processes} \label{4.1}

For a functional linear process $(X_n)_{n\in\Z}$ we focus on the orthogonal projection
\begin{align}
X_{D,n}=P_{A_D}(X_n)=\sum_{j=1}^D \langle X_n , \nu_j \rangle \nu_j, \quad n\in\Z, \label{projX}
\end{align}
for $(\nu_i)_{i\in\N}$ and $A_D$ as in \eqref{AD}.
We will often use the following isometric isomorphism between two Hilbert spaces of the same dimension.

\ble \label{isomor}
Define $A_D$ as in \eqref{AD}. The map $T:A_D\rightarrow \R^D$ defined by $Tx=(\langle x, \nu_i \rangle )_{i=1,\dots,D}$ is a bijective linear mapping with 
$\langle Tx,Ty\rangle_{\R^D} = \langle x,y \rangle$ for all $x,y \in A_D$. 
Hence, $\LCS(F_{D,n})$ is isometrically isomorphic to $\spa\{\mathbf{X}_{D,1},\dots,\mathbf{X}_{D,n}\}$.
Moreover, $(X_{D,n})_{n\in\Z}$ as defined in \eqref{projX} is isometrically isomorphic to the $D$-dimensional vector process  
\beam\label{dvector}
\mathbf{X}_{D,n}:=(\langle X_n, \nu_1\rangle, \dots,\langle X_n, \nu_D \rangle )^\top, \quad n\in\Z.
\eeam 
 
\ele

When choosing $(\nu_i)_{i\in\N}$ as the eigenfunctions of the covariance operator $C_X$ of $(X_n)_{n\in\Z}$, the representation \eqref{projX} is a truncated version of the Karhunen-Lo\'eve decomposition (see e.g. \cite{bosq}, Theorem~1.5).

As known from Example~\ref{exar}, the structure of $(X_n)_{n\in\Z}$ does in general not immediately reveal the dynamics of $(X_{D,n})_{n\in\Z}$. 
Starting with the representation of $(X_{D,n})_{n\in\Z}$ as in \eqref{process} with $\psi_0=I_H$ and using similar notation as in \eqref{ch4_vectorreal}, the  $D$-dimensional vector process $(\mathbf X_{D,n})_{n\in\Z}$ can be written as
\begin{align}\label{bold}
\mathbf X_{D,n}= \mathbf{E}_{D,n} + \sum_{j=1}^\infty \Big(\mathbf{\Psi}_{D,j} \mathbf{E}_{D,n-j}+\mathbf{\Psi}_{D,j}^{\infty}\mathbf{E}_{n-j}^{\infty}\Big),\quad n\in\Z,
\end{align}
where the blocks $\mathbf{\Psi}_{D,j}$, $\mathbf{\Psi}^\infty_{D,j}$, $\mathbf{E}_{D,n}=(\langle \varepsilon_n, \nu_1\rangle, \dots, \langle \varepsilon_n , \nu_D \rangle)^\top$, and $\mathbf{E}_n^\infty = (\langle \varepsilon_n, \nu_{D+1}\rangle, \langle \varepsilon_n , \nu_{D+2} \rangle,\dots)^\top$ are defined analogously to the blocks in \eqref{ch4_bockrep}. Note that this is in general not a vector MA$(\infty)$ representation of a process with innovation $(\mathbf{E}_{D,n})_{n\in\Z}$.

The following proposition summarizes general results  on the structure of $(X_{D,n})_{n\in\Z}$.
Its proof is given in Section~\ref{proofs}.
 
 \begin{proposition} \label{properties}
 Let $(X_n)_{n\in\Z}$ be a stationary (and invertible) functional linear process with WN $(\varepsilon_n)_{n\in\Z}$, such that all eigenvalues of the covariance operator $C_{\varepsilon}$ of $(\varepsilon_n)_{n\in\Z}$ are  positive.
 Then $(X_{D,n})_{n\in\Z}$ is also a stationary (and invertible) functional linear process with some WN $(\tilde{\varepsilon}_{D,n})_{n\in\Z}$. $(\tilde{\varepsilon}_{D,n})_{n\in\Z}$ is isometrically isomorphic to the $D$-dimensional vector process $(\tilde{ \mathbf E}_{D,n})_{n\in\Z}$, defined by $
 \tilde{\mathbf{E}}_{D,n}:=(\langle \tilde{\varepsilon}_{D,n}, \nu_1\rangle, \dots,\langle \tilde{\varepsilon}_{D,n}, \nu_D \rangle )^\top$. Furthermore define  $\mathbf{\calm}_{D,n}=\spa\{\mathbf X_{D,t},-\infty<t\leq n\}$. Then
 \begin{align}\label{Epstilde}
 \tilde{ \mathbf E}_{D,n} = \mathbf{E}_{D,n}  +\mathbf{\Psi}_{D,1}^{\infty}(\mathbf{E}_{n-1}^{\infty} - P_{\mathbf{\calm}_{D,n-1}}(\mathbf{E}^\infty_{n-1})):= \mathbf{E}_{D,n} + \mathbf{\Delta}_{D,n-1}, \quad n\in\Z.
 \end{align} 
The lagged covariance operator $C_{X_D;h}$ of $(X_{D,n})_{n\in\Z}$ is given by
  \begin{align}\label{cxdh}
  C_{X_D;h}= \E[P_{A_D} X_0 \otimes P_{A_D} X_h ] = P_{A_D} \E [X_0\otimes X_h ] P_{A_D} = P_{A_D} C_{X;h} P_{A_D},\quad h\in\Z.
  \end{align}
 \end{proposition}
 
By Lemma~\ref{isomor}, $(X_{D,n})_{n\in\Z}$ is isomorphic to the $D$-dimensional vector process $(\mathbf{X}_{D,n})_{n\in\Z}$ as defined in \eqref{dvector}.
The prediction problem can therefore be solved by methods from multivariate time series analysis.
More precisely, we define for fixed $D\in\N$ 
\begin{align*} 
F_{D,n} =\{X_{D,1},\dots,X_{D,n} \}\quad\mbox{and}\quad  \widehat{X}_{D,n+1} = P_{\LCS(F_{D,n})}(X_{n+1}), 
\end{align*}
i.e., $\widehat{X}_{D,n+1}$ is the best linear functional predictor based on $F_{D,n}$  for $n\in\N$.

We formulate the Multivariate Innovations Algorithm for this setting.

\begin{proposition}[\cite{brockwell}, Proposition~11.4.2] \label{prop_iad}
Let $(X_n)_{n\in\Z}$ be a stationary functional linear process  
and $(X_{D,n})_{n\in\Z}=(P_{A_D}X_n)_{n\in\Z}$ as in \eqref{projX}. 
If $C_{X_D}$ is invertible on $A_D$, then the best linear functional predictor $\widehat{X}_{D,n+1}$ of $X_{n+1}$  based on $\LCS(F_{D,n})$ 
 can be computed by the following set of recursions:
\begin{align}
& \wh{X}_{D,1}=0\quad\mbox{and}\quad V_{D,0}=C_{X_{D};0},\notag\\
&\wh{X}_{D,n+1}= \sum_{i=1}^{n} \theta_{D,n,i} (X_{D,n+1-i}-\wh{X}_{D,n+1-i}), \label{xdhat}\\
&\theta_{D,n,n-i}=\Big(C_{X_{D};n-i} - \sum_{j=0}^{i-1} \theta_{D,n,n-j} \ V_{D,j} \ \theta_{D,i,i-j}^*\Big)V_{D,i}^{-1}, \quad i=1,\dots,n-1, \label{theta} \\
&V_{D,n} =C_{X_{D,n+1}-\wh{X}_{D,n+1}}= C_{X_{D};0} - \sum_{j=0}^{n-1} \theta_{D,n,n-j}V_{D,j}\theta^*_{D,n,n-j} .\label{vd}
\end{align}
\end{proposition}

The recursions can be solved explicitly in the following order: $V_{D,0},\theta_{D,1,1},V_{D,1},\theta_{D,2,2},\theta_{D,2,1}\dots$.
Thus we found a predictor, which is  in contrast to $\widehat{X}_{n+1}$ from \eqref{blp} easy to compute.
However, since we are not using all available information, we loose predictive power. 
To evaluate this loss we bound the prediction error. 
We show that the error bound can be decomposed in two terms. 
One is due to the dimension reduction, and the other to the statistical prediction error of the finite dimensional model.

\begin{theorem}\label{ninfty}
Let $(X_n)_{n\in\mathbb{Z}}$ be a stationary functional linear process with WN $(\varepsilon_n)_{n\in\Z}$ such that all eigenvalues of $C_{\varepsilon}$ are positive. 
Assume furthermore that $C_{X}$ is invertible on $A_D$.
Recall the best linear functional predictor $ \widehat{X}_{n+1}$ from Definition~\ref{defblp}.
\\
(i) Then for all  $n\in\N$ the prediction error is bounded:
\begin{align}
E\Vert X_{n+1} - \widehat{X}_{n+1} \Vert ^2 &\leq E \Vert X_{n+1} - \wh{X}_{D,n+1} \Vert^2 = \sum_{i>D} \langle C_{X}\nu_i, \nu_i \rangle + \Vert  V_{D,n}\Vert_{\caln} ^2 .  \label{theo1i}
\end{align}
(ii) If additionally  $(X_n)_{n\in\mathbb{Z}}$ is invertible, then
\begin{align*}
&\lim_{n\rightarrow\infty}E\Vert X_{n+1} - \wh{X}_{D,n+1} \Vert^2 = \sum_{i>D} \langle C_{X}\nu_i, \nu_i \rangle + \Vert C_{\tilde{\varepsilon}_{D}} \Vert_{\caln}^2.
\end{align*}
\end{theorem}

\begin{proof}
(i) Since $\wh{X}_{D,n+1}=P_{\LCS(F_{D,n})}(X_{n+1})$ and $\widehat{X}_{n+1}=P_{\LCS({F_n})}(X_{n+1})$, and since $ \LCS(F_{D,n}) \subseteq \LCS({F_n})$, the first inequality follows  immediately from the projection theorem.
Furthermore, since $X_{n+1} - X_{D,n+1}\in A_D^{\bot}$ (the orthogonal complement of $ A_D$) and $X_{D,n+1},  X_{D,n+1} -\wh{X}_{D,n+1} \in A_D$, we have $\langle X_{n+1} - X_{D,n+1},  X_{D,n+1} -\wh{X}_{D,n+1} \rangle =0$. 
Therefore,
\begin{align*}
\E\Vert X_{n+1} - \wh{X}_{D,n+1} \Vert^2 &= \E\Vert X_{n+1} - X_{D,n+1}  +  X_{D,n+1}- \wh{X}_{D,n+1} \Vert^2\\
&= \E\Vert X_{n+1} - X_{D,n+1}\Vert ^2   + \E\Vert  X_{D,n+1} -\wh{X}_{D,n+1} \Vert^2.
\end{align*}
By \eqref{traceeq} we have $ \E\Vert  X_{D,n+1} -\wh{X}_{D,n+1} \Vert^2 = \Vert\E[(  X_{D,n+1} -\wh{X}_{D,n+1} )\otimes ( X_{D,n+1} -\wh{X}_{D,n+1})]\Vert_{\caln}$, which is equal to $\Vert V_{D,n}\Vert_{\caln}$ by \eqref{vd}. 
Furthermore, 
\begin{align*}
 \E\Vert X_{n+1} - X_{D,n+1}\Vert ^2 &= \E\big\langle \sum_{i>D} \langle X_{n+1},\nu_i\rangle \nu_i, \sum_{j>D} \langle X_{n+1},\nu_j\rangle \nu_j \big\rangle\\
 &= \sum_{i,j>D} \E \big\langle X_{n+1}\langle X_{n+1} ,\nu_i\rangle, \nu_j \big\rangle\langle \nu_i,\nu_j\rangle\\
 &=\sum_{i>D}\langle C_X\nu_i,\nu_i\rangle.
\end{align*}
\noindent (ii) By (i), what is left to show is that $ \Vert  V_{D,n}\Vert_{\caln} ^2 \rightarrow  \Vert C_{\tilde{\varepsilon}_{D}} \Vert_{\caln}^2$ for $n\rightarrow \infty$. 
However, this is an immediate consequence of the Multivariate Innovations Algorithm under the assumption that $(X_{D,n})_{n\in\Z}$ is invertible (see Remark~4 in Chapter~11 of \cite{brockwell}). 
However, invertibility of $(X_{D,n})_{n\in\Z}$ is given by Proposition~\ref{properties}, which finishes the proof.
\end{proof}

The above theorem states that for a stationary, invertible functional linear process,
for increasing sample size the predictor restricted to the $D$-dimensional space performs arbitrarily well in the sense that in the limit only the statistical prediction error remains. 
However, our goal in \eqref{goal} is not satisfied. The dimension reduction induces the additional error term $\sum_{i>D}\langle C_X(\nu_i),\nu_i\rangle $ independently of the sample size.
If $A_D$ is spanned by eigenfunctions of the covariance operator $C_X$ with eigenvalues $\lambda_i$, the prediction error due to dimension reduction  is $\sum_{i>D}\lambda_i$.  

We now investigate the special case of functional moving average processes with finite order.

\subsection{Prediction of \FMA$(q)$} \label{4.2}

\FMA$(q)$ processes for $q\in\N$ as in Definition~\ref{FMA} are an interesting and not very well studied class of functional linear processes. We start with the \FMA(1) process as example.

\bexam\label{exma}
Consider a FMA(1) process $(X_n)_{n\in\Z}$ defined by
$$X_n= \psi \varepsilon_{n-1} + \varepsilon_n,\quad n\in\Z,$$ 
for some $\psi\in \call$ and WN $(\varepsilon_n)_{n\in\Z}$.
The representation of \eqref{bold} reduces to
 \begin{align*}
 \mathbf{X}_{D,n}= \mathbf{\Psi}_D \mathbf{E}_{D,n-1} + \mathbf{\Psi}_D^{\infty}\mathbf{E}_{n-1}^{\infty} + \mathbf{E}_{D,n}, \quad n\in\Z.
 \end{align*}  
 As $\mathbf{X}_{D,n}$ depends on $ \mathbf{E}_{D,n-1}$, $\mathbf{E}_{n-1}^{\infty}$ and  $\mathbf{E}_{D,n}$, it is in general not a vector MA$(1)$ process with WN $(\mathbf{E}_{D,n})_{n\in\Z}$. 
 \halmos
\eexam

However, we can state the dynamics of a finite dimensional subprocess of a \FMA$(q)$ process.

\begin{theorem}\label{th4}
Let $(X_n)_{n\in\Z}$ be a stationary \FMA$(q)$ process for $q\in\N$ and $A_D$ be as in \eqref{AD}. 
Then $(X_{D,n})_{n\in\Z}=(P_{A_D}X_n)_{n\in\Z}$ as defined in \eqref{projX} is a stationary \FMA$(q^*)$ process for $q^*\leq q$ satisfying
\begin{align}
X_{D,n}=\sum_{j=1}^{q^*}\tilde\psi_{D,j}  \tilde{\varepsilon}_{D,n-j} + \tilde{\varepsilon}_{D,n}, \qquad n\in \Z, \label{MAK}
\end{align}
where $\tilde\psi_{D,j} \in \call$ for $j=1,\dots,q^*$ and $(\tilde{\varepsilon}_{D,n})_{n\in\Z}$ is WN.
Moreover,  $(\tilde{\varepsilon}_{D,n})_{n\in\Z}$ is isometrically isomorphic to $(\tilde{\mathbf{E}}_{D,n})_{n\in\Z}$  as defined  in \eqref{Epstilde}.
If $q^*=0$, then $(X_{D,n})_{n\in\Z}$ is WN.
\end{theorem}

\begin{proof}
By Proposition~\ref{properties} $(X_{D,n})_{n\in\Z}$ is stationary. 
Furthermore, by \eqref{cxdh} and Proposition~\ref{maq} $C_{X_{D};h}=P_{A_D} C_{X;h} P_{A_D}=0$ for $h>q$, since $C_{X;h}=0 $ for $h>q$.
Hence, again by Proposition~\ref{maq} $(X_{D,n})_{n\in\Z}$ is a $\FMA(q^*)$ process, where $q^*$ is the largest lag  $j\leq q$ such that $C_{X_D;j}=P_{A_D} C_{X;j} P_{A_D} \neq 0$. 
Thus, \eqref{MAK} holds for some linear operators $\tilde\psi_{D,j}\in  \call$ and WN $(\tilde{\varepsilon}_{D,n})_{n\in\Z}$. 
The fact that  $(\tilde{\varepsilon}_{D,n})_{n\in\Z}$ is isometrically isomorphic to $(\tilde{\mathbf{E}}_{D,n})_{n\in\Z}$ as in \eqref{Epstilde} is again a consequence of the Wold decomposition of $(X_{D,n})_{n\in\Z}$ and follows from the proof of Proposition~\ref{properties}.
\end{proof}

The fact that every subprocess of a \FMA$(q)$ is a \FMA$(q^*)$ with $q^*\leq q$ simplifies the algorithm of Proposition~\ref{prop_iad}. 
Since $C_{X_D;h}=0$ for $\vert h \vert >q$ modifies \eqref{xdhat}-\eqref{vd} as follows: for $n>q^*$,
\begin{align*}
&\wh{X}_{D,n+1}= \sum_{i=1}^{q^*} \theta_{D,n,i} (X_{D,n+1-i}-\wh{X}_{D,n+1-i}) \\
&\theta_{D,n,k}=\Big(C_{X_{D};k} - \sum_{j=0}^{n-k-1} \theta_{D,n,n-j} \ V_{D,j} \ \theta_{D,n-k,n-k-j}^*\Big)V_{D,n-k}^{-1}, \quad k=1,\dots,q^*,  \\
&V_{D,n} =C_{X_{D,n+1}-\wh{X}_{D,n+1}}= C_{X_{d};0} - \sum_{j=1}^{q^*} \theta_{D,n,j}V_{D,n-j}\theta^*_{D,n,j} .
\end{align*}

We now investigate the prediction error $E \Vert X_{n+1} - \wh{X}_{D,n+1}  \Vert^2$ of Theorem~\ref{ninfty} for a functional linear process.
For $D\rightarrow \infty$, obviously, $\sum_{i>D} \langle C_{X_0}(\nu_i), \nu_i \rangle \rightarrow 0$. 
However, the second term $\Vert V_{D,n}\Vert_{\caln}$ on the rhs of \eqref{theo1i} is not defined in the limit, since the inverse of $V_{D,j}$ in \eqref{theta} is no longer bounded when $D\rightarrow\infty$. 
To see this, take e.g. $V_{D,0}$. By \eqref{vd}, since $\wh{X}_{D,1}=0$ and since $(X_{D,n})_{n\in\Z}$ is stationary, $$V_{D,0}=C_{X_{D,1}-\wh{X}_{D,1}}=C_{X_{D,1}}=C_{X_D}.$$  
By \eqref{cxdh} for $h=0$ we find $C_{X_D}=P_{A_D}C_XP_{A_D}$, hence for all $x\in H$, $\Vert (C_X-C_{X_D})(x)\Vert \rightarrow 0$ for $D\rightarrow\infty$. 
But, since $C_X$ is not invertible on the entire of $H$, neither is $\lim_{D\rightarrow\infty}C_{X_D}$. Therefore, $\lim_{D\rightarrow\infty}\wh{X}_{D,n+1}$ is not defined.

To resolve this problem, we propose a tool used before in functional data analysis, for instance in \cite{bosq} for the estimation of \FAR$(1)$. 
We increase the dimension $D$ together with the sample size $n$ by choosing $d_n:=D(n)$ and $d_n\rightarrow\infty$ with $n\rightarrow \infty$. 
However, since the Innovations Algorithm is based on a recursion, it will always start with $V_{d_n,0}=C_{X_{d_n}}$, which again is not invertible on $H$ for $d_n\rightarrow\infty$. 
For the Innovations Algorithm we increase $D$ iteratively such that $V_{d_1,0}$ is inverted on say $A_1$, $V_{d_2,1}$ is inverted on $A_2, \dots$ and so on. 
To quantify a convergence rate in Theorem~\ref{dinfty} below we restrict ourselves to projections on eigenspaces of the covariance operator $C_X$ of the underlying process.

\section{Prediction based on projections of increasing subspaces of $H$} \label{sec5}

In this section we propose a functional version of the Innovations Algorithm. 
Starting with the same idea as in Section \ref{sec4}, we project the functional data on a finite dimensional space. 
However, we now let the dimension of the space on which we project depend on the sample size. 
More precisely, let $(X_n)_{n\in \Z}$ be a stationary functional linear process  with covariance operator $C_X$. 
For some positive, increasing sequence $(d_n)_{n\in\N}$ in $\N$ such that $d_n\rightarrow\infty$ with $n\rightarrow\infty$, we define
\begin{align}\label{bdn}
	(\nu_i)_{i\in\N}\quad\mbox{and}\quad A_{d_n}=\spa\{\nu_1,\dots,\nu_{d_n}\}, \quad n\in\N, 
\end{align}
where $(\nu_i)_{i\in\N}$ are now chosen as the eigenfunctions of $C_X$, and $(A_{d_n})_{n\in\N}$ is an increasing sequence of subspaces of $H$. 
Instead of applying the Innovations Algorithm to $(P_{A_d}X_1,\dots,P_{A_d}X_n)$ as in Proposition~\ref{prop_iad}, we apply it now to $(P_{A_{d_1}}X_1,\dots,P_{A_{d_n}}X_n)$.  

\begin{proposition} \label{prop_iadn}
	Let $(X_n)_{n\in\Z}$ be a stationary functional linear process with covariance operator $C_X$ with eigenpairs $(\lambda_j,\nu_j)_{j\in\N}$, where $\lambda_j>0$ for all $j\in\N$. 
	Let $(d_n)_{n\in\N}$ be a positive sequence in $\N$ such that $d_n\uparrow\infty$ as $n\rightarrow \infty$. 
	 Define furthermore the best linear predictor of $X_{n+1}$ based on $\LCS(F'_{d_n,n})$ for $n\in\N$ as
	 \begin{align}\label{Fdn}
	 F'_{d_n,n}=\{X_{d_1,1},X_{d_2,2},\dots,X_{d_n,n}\}\quad\mbox{and}\quad  \widehat{X}_{d_{n+1},n+1} = P_{\LCS(F'_{d_n,n})}(X_{n+1}).
	 \end{align}
Then  $\widehat{X}_{d_{n+1},n+1}$ is given by the following set of recursions:
	\begin{align}
	&\wh{X}_{d_1,1}=0\quad\mbox{and}\quad V_{d_1,0}=C_{X_{d_1}},\notag\\
	&\wh{X}_{d_{n+1},n+1}= \sum_{i=1}^{n} \theta_{d_{n-i+1},n,i} (X_{d_{n+1-i},n+1-i}-\wh{X}_{d_{n+1-i},n+1-i}) \label{xdhat1},\\
	&\theta_{d_{i+1},n,n-i}=\Big(P_{A_{d_{n+1}}}C_{X;n-i}P_{A_{d_{i+1}}} - \sum_{j=0}^{i-1} \theta_{d_{j+1},n,n-j}  \ V_{d_{j+1},j} \ \theta_{d_{j+1},i,i-j}^*\Big)V_{d_{i+1},i}^{-1}, \quad i=1,\dots,n-1, \label{theta1} \\
	&V_{d_{n+1},n} =C_{X_{d_{n+1},n+1}-\wh{X}_{d_{n+1},n+1}}= C_{X_{d_{n+1}}} - \sum_{j=0}^{n-1} \theta_{d_{j+1},n,n-j}V_{d_{j+1},j}\theta^*_{d_{j+1},n,n-j}. \label{vd1}
	\end{align}
\end{proposition}

\begin{proof}
The proof is based on the proof of Proposition~11.4.2 in \cite{brockwell}. First notice that the representation 
	\begin{align*}
	\wh{X}_{d_{n+1},n+1}= \sum_{i=1}^n \theta_{d_{n-i+1},n,i} (X_{d_{n+1-i},n+1-i}-\wh{X}_{d_{n+1-i},n+1-i}), \quad n\in\N,
	\end{align*}
results from the definition of  $\widehat{X}_{d_{n+1},n+1} = P_{\LCS(F'_{d_n,n})}(X_{n+1})$. Multiplying both sides of \eqref{xdhat1} with $\langle X_{d_{k+1},k+1}-\wh{X}_{d_{k+1},k+1}, \cdot \rangle $ for $0\leq k \leq n$ and taking expectations, we get
	\begin{align*}
 E \big [ (X_{d_{k+1},k+1}-	&\wh{X}_{d_{k+1},k+1}) \otimes \wh{X}_{d_{n+1},n+1} \big]\\
	\qquad & = \sum_{i=1}^n \theta_{d_{n-i+1},n,i} E \big[(X_{d_{k+1},k+1}-\wh{X}_{d_{k+1},k+1} )\otimes (X_{d_{n+1-i},n+1-i}-\wh{X}_{d_{n+1-i},n+1-i})  \big] \\
	& = \theta_{d_{k+1},n,n-k}  E \big[(X_{d_{k+1},k+1}-\wh{X}_{d_{k+1},k+1}) \otimes (X_{d_{k+1},k+1}-\wh{X}_{d_{k+1},k+1})\big] ,
	\end{align*}
where we used that $E\langle  X_{d_{n+1},n+1}-\wh{X}_{d_{n+1},n+1}, X_{d_{k+1},k+1}-\wh{X}_{d_{k+1},k+1} \rangle = 0$ for $k\neq n$. Now with the definition of $V_{d_{k+1},k}$ in \eqref{vd1},
	\begin{align}
	E \big [( X_{d_{k+1},k+1}-\wh{X}_{d_{k+1},k+1})\otimes X_{d_{n+1},n+1} \big]  = \theta_{d_{k+1},n,n-k} V_{d_{k+1},k}. \label{help1}
	\end{align}
By representation \eqref{xdhat1} for $n=k$ and the fact that $V_{d_{k+1},k}$ is finite dimensional and therefore invertible, since all eigenvalues of $C_X$ are   positive,
\begin{align*}
\theta_{d_{k+1},n,n-k}  = \Big(P_{A_{d_{n+1}}}C_{X;n-k}P_{A_{d_{k+1}}} - \sum_{i=1}^{k} 	E \big [( X_{d_{i},i}-\wh{X}_{d_{i},i} )\otimes X_{d_{n+1},n+1} \big] \theta_{d_{i},k,k-i-1} ^* \Big)  V_{d_{k+1},k}^{-1}.
\end{align*}
However, with \eqref{help1} the expectation on the right-hand-side can be replaced by $\theta_{d_{i},n,n+1-i} V_{d_{i},i-1}$, for $i=1,\dots,k$, which leads to
\begin{align*}
\theta_{d_{k+1},n,n-k}  = \Big(P_{A_{d_{n+1}}}C_{X;n-k}P_{A_{d_{k+1}}} - \sum_{i=1}^{k} \theta_{d_{i},n,n+1-i} V_{d_{i},i-1} \theta_{d_{i},k,k-i-1} ^* \Big)  V_{d_{k+1},k}^{-1}.
\end{align*}
Finally, the projection theorem gives
\begin{align*}
V_{d_{n+1},n} =C_{X_{d_{n+1},n+1}-\wh{X}_{d_{n+1},n+1}}
 = C_{X_{d_{n+1}}} - C_{\wh{X}_{d_{n+1},n+1}} =C_{X_{d_{n+1}}}-\sum_{j=0}^{n-1} \theta_{d_{j+1},n,n-j}V_{d_{j+1},j}\theta^*_{d_{j+1},n,n-j}.
\end{align*}
\end{proof}

\brem
Notice that $X_{d_1,1},X_{d_2,2},\dots,X_{d_n,n}$ is not necessarily stationary. 
However, the recursions above can still be applied, since stationarity is not required for the application of the Innovations Algorithm in finite dimensions, see Proposition~11.4.2 in \cite{brockwell}. 
\erem

If $(X_n)_{n\in\Z}$ is invertible, we can derive asymptotics for $\wh{X}_{d_{n+1},n+1}$ as   $d_n \rightarrow \infty$ and $ n\rightarrow\infty$.

\begin{theorem}\label{dinfty}
Let $(X_n)_{n\in\mathbb{Z}}$ be a stationary, invertible functional linear process with WN $(\varepsilon_n)_{n\in\Z}$ such that all eigenvalues of $C_{\varepsilon}$ are  positive. 
Assume furthermore that all eigenvalues $\lambda_j$, $j\in\N$, of $C_{X}$ are  positive.\\
(i) Let $m_n\rightarrow\infty$, $m_n<n$ and $m_n/n\rightarrow 0$ for $n\rightarrow \infty$ and $d_n\rightarrow\infty$ for $n\rightarrow \infty$ be two positive increasing sequences in $\N$. 
Then
\begin{align}
E\Vert X_{n+1} - \wh{X}_{d_{n+1},n+1} -\varepsilon_{n+1} \Vert^2 = O\big(\sum_{j>m_n} \Vert \pi_j \Vert_{\call} +	\sum_{j>d_{n-m_n}} \lambda_j \big) \rightarrow 0, \quad  n\rightarrow \infty.\label{theoi}
\end{align}
(ii) Denote by $C_{\mathbf{X}_{d_n};h}$ the covariance matrix of the subprocess $(\mathbf{X}_{d_n})_{n\in\Z}$ as defined in Lemma~\ref{isomor}. 
Then all eigenvalues of the spectral density matrix $f_{\mathbf{X}_{d_n}}[\omega]:=\frac{1}{2\pi}\sum_{h\in\Z} e^{-ih\omega} C_{\mathbf{X}_{d_n};h}$ for $-\pi<\omega\leq \pi$ are positive.
Denote by $\alpha_{d_n}>0$ the infimum of these eigenvalues. 
If 
	\begin{align}
				\frac{1}{\alpha_{d_n}} \big(\sum_{j>m_n} \Vert \pi_j \Vert_{\call} +	\sum_{j>d_{n-m_n}} \lambda_j \big) \rightarrow 0,\quad   n\rightarrow \infty. \label{assumptionspec}
			\end{align} 
 then for $i=1,\dots,n$ and for all  $x\in H$,
\begin{align*}
\Vert (\theta_{d_n,n,i} - \gamma_i)(x) \Vert\rightarrow 0,\quad   n\rightarrow \infty .
\end{align*}
\end{theorem}

The proof of this Theorem is given in the next Section~\ref{proofs}.

\brem
(a) \, Part $(i)$ of Theorem~\ref{dinfty} requires only that $d_n\rightarrow\infty$ as $n\rightarrow\infty$. 
No rate is required, and we do not impose any coupling condition of $d_n$ with $m_n$. 
The theory would suggest to let $d_n$ increase as fast as possible. 
In practice, when quantities such as the lagged covariance operators of the underlying process have to be estimated, the variance of the estimators of $P_{d_n} C_{X;h} P_{d_n}$ increases with $d_n$. 
In fact, for instance, for the estimation of $\theta_{d_1,1,1}$ the statistician is faced with the inversion of $P_{d_1}C_XP_{d_1}$. 
Small errors in the estimation of small empirical eigenvalues of $P_{d_1}C_XP_{d_1}$ may have severe consequences for the estimation of $\theta_{d_1,1,1}$. 
This suggests a conservative choice for $d_n$.
The problem is similar to the choice of $k_n$ in Chapter~9.2 of \cite{bosq} concerned with the estimation of the autoregressive parameter operator in a \FAR$(1)$. 
The authors propose to choose $k_n$ based on validation of the empirical prediction error.

(b) \, The choice of $m_n$ in \eqref{theoi} allows us to calibrate two error terms: under the restriction that $m_n/n\rightarrow 0$, choosing a larger $m_n$ increases $\sum_{j>d_{n-m_n}} \lambda_j$, the error caused by dimension reduction. Choosing a smaller $m_n$  will on the other hand increase $\sum_{j>m_n} \Vert \pi_j \Vert$.
\erem

\section{Proofs}\label{proofs}

Before presenting a proof of Theorem~\ref{dinfty} we give some notation and auxiliary results.
Recall that throughout $I_H$ denotes the identity operator  on $H$.
We also recall the notation and results provided in Section~\ref{sec2}, which we shall use below without specific referencing.

 Let $(X_n)_{n\in\Z}$ be a stationary functional linear process. 
 Then for $n\in\N$ define the covariance operator of the vector $(X_n,\dots,X_1)$ by
\begin{small}
\begin{align}\label{gammam}
\Gamma_{n}& :=\begin{pmatrix}
\mathbb{E}[X_{n}\otimes X_{n}]  & \mathbb{E}[X_{n}\otimes X_{n-1}] & \dots & \mathbb{E}[X_{n}\otimes X_{1}]\\
\mathbb{E}[X_{n-1}\otimes X_{1}] & \mathbb{E}[X_{n-1}\otimes X_{n-1} ]  & \dots &\vdots \\
\vdots & & \ddots & \\
\mathbb{E}[X_{1}\otimes X_{n}] & \dots & & \mathbb{E}[X_{1}\otimes X_{1}]  
\end{pmatrix} 
 =  \begin{pmatrix}
C_X  &C_{X;1} & \dots &  C_{X;n-1} \\
 C_{X;-1}  &  C_{X}    & \dots &\vdots \\
\vdots & & \ddots & \\
C_{X;-(n-1)} & \dots & &  C_X
\end{pmatrix},
\end{align}
\end{small}
i.e., $\Gamma_{n}$ is an operator acting on $H^n$, where $H^n$ is the Cartesian product of $n$ copies of $H$. 
Recall that $H^n$ is again a Hilbert space, when equipped with the scalar product
\begin{align*}
\langle x, y \rangle_n = \sum_{i=1}^n \langle x_i, y_i \rangle
\end{align*}
(see \cite{bosq}, Section~5 for details). 
As the covariance operator of  $(X_{n},X_{n-1},\dots,X_{1})$, $\Gamma_n$ is  self-adjoint, nuclear, and has the spectral representation (cf. Theorem~5.1 in \cite{gohberg})
\begin{align*}
\Gamma_{n}=\sum_{j=1}^{\infty} \lambda_j^{(n)}  \nu_j^{(n)}\otimes \nu_j^{(n)}, \quad n\in\N,
\end{align*}  
with eigenpairs $(\lambda_j^{(n)},\nu_j^{(n)})_{j\in\N}$. 

Furthermore, define the operators $P_{(d_n)}$ and $P_D$ acting on $H^n$ by 		 
 	\begin{align}
		  	P_{(d_n)}&=\text{diag}\ \big(P_{A_{d_n}},P_{A_{d_{n-1}}},\dots,P_{A_{d_1}}) \quad \text{and}		\quad  	P_{D}=\text{diag}\ \big(P_{A_D},P_{A_D},\dots,P_{A_D}), \label{pdn}
	\end{align}
then 
\begin{align*}
\Gamma_{(d_n),n}:=P_{(d_n)} \Gamma_n P_{(d_n)}\quad\mbox{and}\quad\Gamma_{D,n}:=P_D\Gamma_n P_D.
\end{align*}
Note that $\Gamma_{(d_n),n}$ is in fact the covariance operator of $(X_{d_n,n},\dots,X_{d_1,1})$ and has rank $k_n:=\sum_{i=1}^{n} d_i$, whereas $\Gamma_{D,n}$ is the covariance operator of $(X_{D,n},\dots,X_{D,1}) $ and has rank $D\cdot n$. The operators $\Gamma_{(d_n),n}$ and $\Gamma_{d_n,n}$ are therefore self-adjoint nuclear operators with spectral representations
\beam
\Gamma_{(d_n),n}=\sum_{j=1}^{k_n} \lambda_{(d_n),j}^{(n)} e^{(n)}_{(d_n),j}\otimes e^{(n)}_{(d_n),j} \quad\mbox{and}\quad \Gamma_{d_n,n}=\sum_{j=1}^{d_n\cdot n} \lambda_{d_n,j}^{(n)} e^{(n)}_{d_n,j}\otimes e^{(n)}_{d_n,j}. \label{specrep}
\eeam
 We need the following auxiliary results.

\begin{lemma}[Theorem~1.2 in \cite{mitchell2}]\label{specdens}
Let $(\mathbf{X}_{D,n})_{n\in\Z}$ be a $D-$variate stationary, invertible linear process satisfying $$\mathbf{X}_{D,n}=\sum_{i=1}^{\infty} \mathbf{\Psi}_i \mathbf{E}_{n-i} + \mathbf{E}_n,\quad n\in\Z,$$ 
with $\sum_{i=1}^{\infty}\Vert\mathbf{\Psi}_i\Vert_{2} < \infty$ ($\Vert\cdot\Vert_2$ denotes the Euclidean matrix norm) and  WN $(\mathbf{E}_{D,n})_{n\in\Z}$ in $L^2_{\R^D}$ with non-singular covariance matrix $C_{\mathbf{E}_D}$.
Let $C_{\mathbf{X}_D}$ be the covariance matrix of $ (\mathbf{X}_{D,n})_{n\in\Z}$.
Then  the spectral density matrix $f_{\mathbf{X}_D} [\omega]:=\frac{1}{2\pi}\sum_{h\in\Z} e^{-ih\omega} C_{\mathbf{X}_{D;h}}$ for $-\pi<\omega\leq\pi$ has only positive eigenvalues.
 Let $\alpha_D$ be their infimum. Then the eigenvalues $(\lambda^{(n)}_i)_{i=1,\dots,D\cdot n}$ of $\Gamma_{D,n}$  as in \eqref{gammam} are bounded below as follows:
\begin{align*}
0<2\pi \alpha_D \leq  \lambda^{(n)}_{D\cdot n} \leq \dots \leq \lambda^{(n)}_{1}.
\end{align*}
\end{lemma}

The following is a consequence of the Cauchy-Schwarz inequality.

\ble\label{lehelp2}
For $j,l\in\N$ let  $(\la_j,\nu_j)$ and $(\la_l,\nu_l)$ be eigenpairs of $C_X$. Then for $i,k\in\Z$,
\begin{align}
		\langle C_{X;i-k}\nu_j, \nu_l \rangle 
	&\le \lambda_j ^{1/2} \lambda_l ^{1/2}. \label{help2}
	\end{align}
	\ele

\bproof
With the definition of the lagged covariance operators in \eqref{cxh} and then the Cauchy-Schwarz inequality, we get by stationarity of $(X_n)_{n \in\Z}$
	\begin{align}
	\langle C_{X;i-k}\nu_j, \nu_l \rangle &= \langle \E \langle X_i, \nu_j \rangle X_k, \nu_l \rangle \notag 
	\, = \, \E [\langle X_i, \nu_j \rangle \langle X_k, \nu_l\rangle ]\notag\\
	&\leq \big(\E \langle X_i,\nu_j \rangle ^2 \big)^{1/2} \big(\E \langle X_k,\nu_l \rangle ^2 \big)^{1/2}.\notag
	\end{align}
	We find $\E \langle X_i,\nu_j \rangle ^2= \E\langle\langle X_i,\nu_j\rangle X_i,\nu_j\rangle = \langle C_X\nu_j,\nu_j\rangle=\lambda_j$, which implies \eqref{help2}.
\eproof

So far we only considered the real Hilbert space $H=L^2([0,1])$. 
There is a natural extension to the complex Hilbert space by defining the scalar product $\langle x,y \rangle = \int_{0}^1 x(t) \bar{y}(t) dt$ for complex valued functions $x,y:[0,1]\to \C$. 
As in Section~7.2 of \cite{bosq}, for $(\psi_j)_{j\in\N}\subset \call$  we define the {\em complex} operators
\begin{align}
A[z]:=\sum_{j=0}^{\infty} z^j \psi_j, \quad z\in\C, \label{AK}
\end{align}
such that the series converges in the operator norm.
We need some methodology on frequency analysis of functional time series, recently studied in \cite{panaretros}.
The functional discrete Fourier transform of $(X_1,\dots,X_n)$ is defined by
\begin{align*}
S_n(\omega) = \sum_{j=1}^n X_j e^{-ij\omega}, \quad \omega \in (-\pi,\pi].
\end{align*}
By Theorem~4 of \cite{cerovecki}, for all $\omega \in (-\pi,\pi]$, if  $(X_n)_{n\in\Z}$ is a linear process with $\sum_{i=1}^\infty \Vert \psi_j \Vert_\call < \infty$, then $\frac{1}{\sqrt{n}} S_n(\omega)$ converges in distribution as $n\rightarrow \infty$ to a complex Gaussian random element with covariance operator 
\begin{align*}
 2\pi\calf_X[\omega] :=\sum_{h\in\Z} C_{X;h} e^{-ih\omega}. 
\end{align*} 
The {\em  spectral density operator}  $\calf_X[{\omega}]$ of $(X_n)_{n\in\Z}$ is  non-negative, self-adjoint and nuclear (see Proposition~2.1 in  \cite{panaretros}).  
 
Theorem~1 and 4 of \cite{cerovecki} infer the following duality between $C_{X;h}$ and $\calf_X[{\omega}]$, with $A[z]$ as in \eqref{AK} and adjoint $A[z]^*$:
 \begin{align}\label{eq4.9}
 C_{X;h}=\int_{-\pi}^{\pi} \calf_X[{\omega}] e^{ih\omega} d\omega,\quad h\in\Z
 \quad\mbox{and}\quad 
\calf_X[{\omega}]=\frac{1}{2\pi} A[e^{-i\omega}] C_{\varepsilon} A[e^{-i\omega}]^*,\quad\omega\in(-\pi,\pi].
\end{align}
The following Lemma is needed for the subsequent proofs, but may also be of interest by itself.
\begin{lemma}\label{invspec}
Let $(X_n)_{n\in\Z}$ be a stationary, invertible functional linear process with WN $(\varepsilon_n)_{n\in\Z}$, such that all eigenvalues of $C_{\varepsilon}$ are positive. Then for all $\omega\in(-\pi,\pi]$ the spectral density operator $\calf_X[{\omega}]$ has only positive eigenvalues.
\end{lemma}

\begin{proof}
The proof is an extension of the proof of Theorem~3.1 in \cite{nsiri} to the infinite dimensional setting.
Define for $A[z]$ as in \eqref{AK} and $(\pi_i)_{i\in\N}$ as in \eqref{invertible}
\begin{align*}
P[z]:= \sum_{j=0}^{\infty} z^j\pi_j\quad \text{and} \quad
D[z] := P[z]A[z],\quad z\in\C.
\end{align*}
Since  $A[z]$ and $P[z]$ are power series, also $D[z]$ can be represented by
\begin{align*}
D[z]= \sum_{j=0}^{\infty} z^j\delta_j, \quad z\in\C,
\end{align*}
for $\delta_j\in\call$. 
Let $B$ be the backshift operator. 
Then $X_n=A[B] \varepsilon_n$ and $\varepsilon_n=P[B]X_n$; in particular, 
\begin{align}
\varepsilon_n=P[B] X_n = P[B]A[B]\varepsilon_n=D[B]\varepsilon_n, \qquad n\in\Z. \label{PA}
\end{align}
Since all eigenvalues of $C_{\varepsilon}$ are positive, by equating the coefficients in \eqref{PA}, $D[z]=I_H$ for all $z\in\C$.

Assume that there exists some non-zero $v\in H$ such that $\calf_X[{\omega}](v)=0$. 
Then by \eqref{eq4.9}, 
$$\frac{1}{2\pi} A[e^{-i\omega}] C_{\varepsilon} A[e^{-i\omega}]^*(v)=0.$$ 
But since all eigenvalues of $C_{\varepsilon}$ are positive, there exists some non-zero $u\in H$ such that $A[e^{i\omega}](u)=0$. 
However, since $D[z]=P[z]A[z]=I_H$ for all $z\in\C$, this is a contradiction, and $\calf_X[{\omega}]$ can only have  positive eigenvalues for all $\omega\in(-\pi,\pi]$.
\end{proof}	

\subsection{Proof of Proposition~\ref{properties}}

Stationarity of $(X_{D,n})_{n\in\Z}$ follows immediately from stationarity of $(X_{n})_{n\in\Z}$, since $P_{A_D}$ is a linear shift-invariant transformation. 
The functional Wold decomposition (e.g. Definition~3.1 in \cite{bosq2}) gives a representation of $(X_{D,n})_{n\in\Z}$  as a linear process with WN, say $(\tilde\varepsilon_n)_{n\in\Z}$ in $L^2_H$.
By Lemma~\ref{isomor}, $(X_{D,n})_{n\in\Z}$ is isometrically isomorphic to the vector process $(\mathbf{X}_{D,n})_{n\in\Z}$ as in \eqref{dvector}. 
Analogously, $(\tilde{ \mathbf E}_{D,n})_{n\in\Z}$ defined by $
 \tilde{\mathbf{E}}_{D,n}:=(\langle \tilde{\varepsilon}_{D,n}, \nu_1\rangle, \dots,\langle \tilde{\varepsilon}_{D,n}, \nu_D \rangle )^\top$ is isometrically isomorphic to $(\tilde\varepsilon_{D,n})_{n\in\Z}$.
We give a representation of $(\tilde{ \mathbf E}_{D,n})_{n\in\Z}$.

Define  $\mathbf{\calm}_{D,n}=\spa\{\mathbf X_{D,t},-\infty<t\leq n\}$. 
Then from the multivariate Wold decomposition the WN of $(\mathbf X_{D,n})_{n\in\Z}$ in $L^2 _{\R^D}$ is defined by 
\begin{align}
\tilde{\mathbf E}_{D,n} = \mathbf X_{D,n}- P_{\mathbf{\calm}_{D,n-1}}(\mathbf{X}_{D,n}),\quad n\in\Z. \label{epshelp}
\end{align}
Now recall \eqref{bold} in the following form
$$\mathbf X_{D,n}= \mathbf{E}_{D,n} + \sum_{j=1}^\infty \mathbf{\Psi}_{D,j} \mathbf{E}_{D,n-j}
+\mathbf{\Psi}_{D,1}^{\infty}\mathbf{E}_{n-1}^{\infty}
+ \sum_{j=2}^\infty \mathbf{\Psi}_{D,j}^{\infty}\mathbf{E}_{n-j}^{\infty},\quad n\in\Z,
$$
and we apply the projection operator to all terms.
Firstly, $P_{\mathbf{\calm}_{D,n-1}}(\mathbf{E}_{D,n})=0$, and $\mathbf{E}_{D,n-j}$ and $\mathbf{E}^{\infty}_{n- j-1}$ belong to $\mathbf{\calm}_{D,n-1}$ for all $j\geq 1$.
Hence,
 \begin{align*}
 P_{\mathbf{\calm}_{D,n-1}}(\mathbf X_{D,n})= \sum_{i=1}^{\infty} \mathbf{\Psi}_{D,j} \mathbf{E}_{D,n-j} + \sum_{i=2}^{\infty}\mathbf{\Psi}_{D,j}^{\infty}\mathbf{E}_{n-j}^{\infty}+\mathbf{\Psi}_{D,1}
  P_{\mathbf{\calm}_{D,n-1}}(\mathbf{E}^\infty_{n-1}), \quad n\in\Z,
 \end{align*}
which together with \eqref{epshelp} implies \eqref{Epstilde}.

We now show that $( \mathbf X_{D,n})_{n\in\Z}$ is invertible. The Wold decomposition gives the following representation
\begin{align} \label{multlinproc}
 \mathbf X_{D,n}=\sum_{i=1}^{\infty} \tilde{\mathbf{\Psi}}_{D,i}( \tilde{\mathbf E}_{D,n-i})+ \tilde{\mathbf E}_{D,n},\quad n\in\Z
 \end{align}
 for appropriate $\tilde{\mathbf{\Psi}}_{D,i}$ and innovation process as in \eqref{epshelp}.
Theorem~1 of \cite{nsiri} gives conditions for the invertibility of the stationary  $D$-variate linear process $( \mathbf X_{D,n})_{n\in\Z}$ satisfying \eqref{multlinproc}.

We verify these conditions one by one. \\
(1) \, We start by showing that for all $\omega\in (-\pi,\pi]$ the matrix $\calf_{\mathbf{X}_D}[\omega]$ is invertible, equivalently,  $\langle \calf_{\mathbf{X}_D}[\omega]\mathbf{x},\mathbf{x}\rangle_{\R^D}>0$ for all non-zero $\mathbf{x}\in\R^D$.
By the isometric isomorphy between $\R^D $ and $A_D$ from Lemma~\ref{isomor} we have
\begin{align*}
\big\langle  \calf_{\mathbf{X}_D}[\omega]\mathbf{x},\mathbf{x}\big\rangle_{\R^D} = \big\langle  \calf_{X_D}[\omega]x,x\big\rangle.
\end{align*}
By \eqref{cxdh} the spectral density operator $\calf_{X_D}[{\omega}]$  of $(X_{D,n})_{n\in\Z}$ satisfies
\begin{align}\label{project}
\calf_{X_D}[{\omega}] &= \frac{1}{2\pi}\sum_{h \in \Z} C_{X_D;h} e^{-ih\omega}=\frac{1}{2\pi} \sum_{h\in\Z} P_{A_D}C_{X;h}P_{A_D}e^{-ih\omega}\notag\\
&=P_{A_D}\Big(\frac{1}{2\pi} \sum_{h\in\Z} C_{X;h}e^{-ih\omega}\Big)P_{A_D} =P_{A_D} \calf_X[{\omega}] P_{A_D}. 
\end{align}
However, since $(X_n)_{n\in\Z}$ is invertible, by Lemma~\ref{invspec} all eigenvalues of $\calf_X[{\omega}]$ are positive for all $\omega\in(-\pi,\pi ]$.
Using first \eqref{project}, then that $x\in A_D$ and finally that all eigenvalues of $\calf_X[{\omega}]$ are positive, we get
\begin{align*}
\big\langle  \calf_{X_D}[\omega]x,x\big\rangle = \big\langle  P_{A_D} \calf_X[{\omega}] P_{A_D} x,x\big\rangle  =  \big\langle   \calf_X[{\omega}]  x,x\big\rangle >0.
\end{align*}
Hence, $\langle \calf_{\mathbf{X}_D}[\omega]\mathbf{x},\mathbf{x}\rangle_{\R^D}>0$  and thus $\calf_{X_D}[{\omega}]$ is invertible. \\
(2) \,  We next show that the covariance matrix $C_{\tilde{\mathbf{E}}_D}$ of $(\tilde{\mathbf{E}}_{D,n})_{n\in\Z}$ as in \eqref{Epstilde}  is invertible. 
Since $\mathbf{E}_{D,n}$ and $ \mathbf{\Delta}_{D,n-1}$ from \eqref{Epstilde} are uncorrelated, 
 $C_{\tilde{\mathbf{E}}_{D}}= C_{\mathbf{E}_{D}} + C_{\mathbf{\Delta}_{D}}$. 
 All eigenvalues of $C_{\varepsilon}$ are positive by assumption.
 For all $x\in A_{D}$ we get $\langle x, C_{\varepsilon} x \rangle = \langle \mathbf{x}, C_{\mathbf E_{D}} \mathbf{x} \rangle_{\R^d}$ where $x$ and $\mathbf{x}$ are related by the isometric isomorphism $T$ of Lemma~\ref{isomor}. With the characterization of the eigenvalues of a self-adjoint operator via the Rayleigh quotient as in Theorem~4.2.7 in \cite{hsing}, all eigenvalues of $C_{\mathbf E_{D}}$ are positive.  
Therefore, all eigenvalues of  $C_{\tilde{\mathbf E}_{D}}= C_{\mathbf E_{D}} +C_{\mathbf{\Delta}_{D}} $ are positive, and $C_{\tilde{\mathbf E}_{D}}$ is invertible. \\
(3) \, Finally, summability in Euclidean matrix norm of the matrices $\tilde{\mathbf{\Psi}}_{D,i}$ over $i\in\N$ follows from the properties of the Wold decomposition (see e.g. Theorem~5.7.1 in \cite{brockwell}) and from the summability of $\Vert \psi_i \Vert_\call$ over $i\in\N$. 

 Therefore, all conditions of Theorem~1 of \cite{nsiri} are satisfied and $(\mathbf{X}_{D,n})_{n\in\Z}$ is invertible.
	\halmos

\subsection{Proof of Theorem~\ref{dinfty}~(i)}

	First note that by the projection theorem (e.g. Theorem~2.3.1 in \cite{brockwell}), 
	\begin{align}\label{proof8}
	\E \Vert   X_{n+1}-\wh{X}_{d_{n+1},n+1} \Vert ^2 \leq \E \Vert X_{n+1}- \sum_{i=1}^n \eta_i X_{d_{n+1-i},n+1-i} \Vert^2, \quad n\in\N, 
	\end{align}
	for all $\eta_i \in \call$, $i=1,\dots,n$. 
	Hence, \eqref{proof8} holds in particular for $\eta_i=\pi_i$ for $i=1,\dots,n$, where $\pi_i$ are the operators in the inverse representation of $(X_n)_{n\in\Z}$ of \eqref{invertible}. 
	 Furthermore, by the orthogonality of $\varepsilon_{n+1}$ and $X_k$ for $k<n+1$ and $n\in\N$,
	\begin{align}
	\E \Vert   X_{n+1}-\wh{X}_{d_{n+1},n+1}\Vert ^2 &=
	\E \Vert   X_{n+1}-\wh{X}_{d_{n+1},n+1} - \varepsilon_{n+1} \Vert ^2 +  \E\Vert \varepsilon_{n+1} \Vert ^2\notag
	\end{align}
	Now \eqref{proof8} with $\eta_i=\pi_i$ and then the invertibility of $(X_n)_{n\in\Z}$ yield
	\begin{align}
	\E \Vert   X_{n+1}-\wh{X}_{d_{n+1},n+1} - \varepsilon_{n+1} \Vert ^2 &\leq \E \Vert X_{n+1}- \sum_{i=1}^n \pi_i X_{d_{n+1-i},n+1-i} \Vert^2- \E\Vert \varepsilon_{n+1} \Vert ^2\notag\\
	&=  \E \Vert \sum_{i=1}^{\infty} \pi_i X_{n+1-i } + \varepsilon_{n+1} - \sum_{i=1}^n \pi_i X_{d_{n+1-i},n+1-i} \Vert^2- \E\Vert \varepsilon_{n+1} \Vert ^2\notag\\
	&=\E \Vert \sum_{i=1}^n \pi_i ( X_{n+1-i}-X_{d_{n+1-i},n+1-i}) + \varepsilon_{n+1} +   \sum_{i>n} \pi_i X_{n+1-i}\Vert^2- \E\Vert \varepsilon_{n+1} \Vert ^2\notag.
	\end{align}
	Again by the orthogonality of $\varepsilon_{n+1}$ and $X_k$, for $k<n+1$, since $X_{d_n,n}=P_{A_{d_n}}X_n$, and then using that for $X,Y\in L^2_H$,  $\E\Vert X + Y \Vert ^2 \leq 2\E \Vert X\Vert^2+ 2\E\Vert Y \Vert ^2$,  we get
	\begin{align}\label{proof0}
	\E \Vert   X_{n+1}-\wh{X}_{d_{n+1},n+1} - \varepsilon_{n+1} \Vert ^2 &\leq  \E \Vert \sum_{i=1}^n \pi_i ( I_H-P_{A_{d_{n+1-i}}})X_{n+1-i}+\sum_{i>n} \pi_i X_{n+1-i}  \Vert^2 \notag\\
	&\leq  2 \E \Vert \sum_{i=1}^n \pi_i ( I_H-P_{A_{d_{n+1-i}}})X_{n+1-i} \Vert^2+ 2\E \Vert \sum_{i>n} \pi_i X_{n+1-i} \Vert ^2 \\
	& =: 2 J_1 + 2 J_2.\notag
	\end{align}
	We consider the two terms in \eqref{proof0} separately. 
	From \eqref{traceeq} we get for the first term in \eqref{proof0}
	\begin{align}
	J_1 &= \Big\Vert  \E \Big [\sum_{i=1}^n \pi_i ( I_H-P_{A_{d_{n+1-i}}})X_{n+1-i} \otimes \sum_{i=1}^n \pi_i ( I_H-P_{A_{d_{n+1-i}}})X_{n+1-i} \Big] \Big\Vert_{\caln}.\notag
	\end{align}
Using the triangle inequality together with properties of the nuclear operator norm given in Section~\ref{sec2}, and then the definition of $C_{X;h}$ in \eqref{cxh},
	\begin{align}
	J_1 &\leq \sum_{i,j=1}^n  \Vert \pi_i\Vert_{\call} \Vert \pi_j\Vert_{\call} \big\Vert \E \big[ ( I_H-P_{A_{d_{n+1-i}}})X_{n+1-i} \otimes  ( I_H-P_{A_{d_{n+1-j}}})X_{n+1-j}\big] \big\Vert_{\caln} \notag\\
	&= \sum_{i,j=1}^n  \Vert \pi_i\Vert_{\call} \Vert \pi_j\Vert_{\call} \big\Vert ( I_H-P_{A_{d_{n+1-i}}})C_{X;i-j} ( I_H-P_{A_{d_{n+1-j}}}) \big\Vert_{\caln} := \sum_{i,j=1}^n  \Vert \pi_i\Vert_{\call} \Vert \pi_j\Vert_{\call}  K(i,j).\label{proof1}
	\end{align}
	 By the definition of $A_d$ in \eqref{bdn} and, since by \eqref{projX} we have  $( I_H-P_{A_{d_{i}}})= \sum_{l>d_{i}} \nu_l\otimes\nu_l$, 
	\begin{align*}
	K(i,j) &=  \big\Vert (\sum_{l'>d_{n+1-i}} \nu_{l'}\otimes \nu_{l'}) C_{X;i-j}(\sum_{l>d_{n+1-j}} \nu_l\otimes \nu_l)\big\Vert_{\caln}
	=\big\Vert\sum_{l'>d_{n+1-i}} \sum_{l>d_{n+1-j}} \langle C_{X;i-j} (\nu_l), \nu_{l'} \rangle \nu_l\otimes \nu_{l'}\big\Vert_{\caln}.
	\end{align*}
	With Lemma~\ref{lehelp2}, the definition of the nuclear norm given in Section~\ref{sec2} and the orthogonality of the $(\nu_i)_{i\in\N}$, we get
 	\begin{align}
	K(i,j) &\leq \big\Vert\sum_{l'>d_{n+1-i}} \sum_{l>d_{n+1-j}} \lambda_l^{1/2}\lambda_{l'}^{1/2} \nu_l\otimes \nu_{l'}\big\Vert_{\caln} \notag\\
	&= \sum_{k=1}^{\infty} \big\langle \sum_{l'>d_{n+1-i}} \sum_{l>d_{n+1-j}} \lambda_l^{1/2}\lambda_{l'}^{1/2} \nu_l\otimes \nu_{l'} (\nu_k), \nu_k \big\rangle \notag\\
	&= \sum_{k>\max(d_{n+1-j},d_{n+1-i})} \lambda_k \ \leq \sum_{k>d_{n+1-j}} \lambda_k. \label{proof2}
	\end{align}
	Plugging \eqref{proof2} into \eqref{proof1}, and recalling that $\sum_{i=1}^{\infty} \Vert\pi_i\Vert_{\call}=:M_1<\infty$, we conclude
	\begin{align}
	J_2 &\leq M_1 \sum_{j=1}^n \Vert \pi_j \Vert_{\call} \sum_{l>d_{n+1-j}} \lambda_l. \label{proof6}
	\end{align}
	Now for some $m_n<n$,
	\begin{align}
	\sum_{j=1}^n \Vert \pi_j \Vert_{\call} \sum_{l>d_{n+1-j}} \lambda_l &= \sum_{j=1}^{m_n} \Vert \pi_j \Vert_{\call} \sum_{l>d_{n+1-j}} \lambda_l \ + \sum_{j=m_n}^{n} \Vert \pi_j \Vert_{\call} \sum_{l>d_{n+1-j}} \lambda_l.\label{proof3}
	\end{align}
	Since  $\sum_{j=1}^{m_n} \Vert \pi_j \Vert_{\call}\leq \sum_{j=1}^{\infty}\Vert \pi_j \Vert_{\call}=M_1<\infty$, the first term on the rhs of \eqref{proof3} can be bounded by
		\begin{align}
		\sum_{j=1}^{m_n} \Vert \pi_j \Vert_{\call} \sum_{l>d_{n+1-j}} \lambda_l \leq M_1 \sum_{l>d_{n+1-m_n}} \lambda_l. \label{result1}
			\end{align}
	Furthermore, since $\sum_{l>d_{n+1-j}} \lambda_l \leq \sum_{l=1}^{\infty} \lambda_l = \Vert C_X \Vert_{\caln}<\infty$, the second term of the rhs in \eqref{proof3} can be bounded by
		\begin{align}
		\sum_{j=m_n}^{n} \Vert \pi_j \Vert_{\call} \sum_{l>d_{n+1-j}} \lambda_l &\leq  \Vert C_X \Vert_{\caln} \sum_{j=m_n}^{n} \Vert \pi_j \Vert_{\call} . \label{result3}
		\end{align}
	Hence, from \eqref{proof6} together with \eqref{proof3}, \eqref{result1} and \eqref{result3} we obtain 
	\begin{align}
	J_1 = O \big( \sum_{j=m_n}^{n} \Vert \pi_j \Vert_{\call} + \sum_{l>d_{n+1-m_n}} \lambda_l \big). \label{result4}
	\end{align}
  Concerning $J_2$, the second term of \eqref{proof0} with \eqref{otimesprop}, and then the definition of $C_{X;h}$ in \eqref{cxh} 	yield
  \begin{align*}
		J_2 &= \E \Vert \sum_{i>n} \pi_i X_{n+1-i} \Vert ^2 \, = \, \big\Vert \E\big[ \sum_{i>n} \pi_i X_{n+1-i} \otimes  \sum_{j>n} \pi_j X_{n+1-j} \big]\big\Vert_{\caln}\\
		&=\Vert\sum_{i,j>n} \pi_i C_{X;i-j} \pi_j^* \Vert_{\caln}
		\leq \sum_{i,j>n} \Vert \pi_i\Vert_{\call} \Vert \pi_j\Vert_{\call} \Vert C_{X;i-j} \Vert_{\caln}.
		\end{align*}
		Since $C_{X;i-j}\in\caln$ for all $i,j\in\N$, $\Vert C_{X;i-j} \Vert_{\caln}=:M_2<\infty$, and for some $m_n<n$,
		\begin{align}
		J_2 &\leq M_2 \big(\sum_{i>n}\Vert \pi_i \Vert_{\call}\big)^2 = O \big(\sum_{i>m_n}\Vert \pi_i \Vert_{\call}\big). \label{result2}
		\end{align}
	Finally the combination of \eqref{proof0}, \eqref{result4} and \eqref{result2}  yields assertion (i).
	\halmos

\subsection{Proof of Theorem~\ref{dinfty}~(ii)}
	
Note first that by the projection theorem there is an equivalent representation of $\wh{X}_{d_{n+1},n+1}$ to \eqref{xdhat1} given by
	\begin{align}
		\wh{X}_{d_{n+1},n+1} = P_{\LCS(F'_{d_n,n})}(X_{n+1}) = \sum_{i=1}^{n} \beta_{d_{n+1-i},n,i} X_{d_{n+1-i},n+1-i} \label{alternative}
	\end{align}
	for $F'_{d_n,n}$ as in \eqref{Fdn} and $\beta_{d_{n+1-i},n,i}\in\call$ for $i=1,\dots,n$.
	Furthermore, for $k=1,\dots,n$, we define 
	the best linear predictor of $X_{n+1}$ based on $F'_{d_n,n}(k)=\{X_{d_{n+1-k},n+1-k},X_{d_{n-k+2},n+2-k},\dots,X_{d_n,n}\}$ by
	\begin{align}
			\wh{X}_{d_{n+1},n+1}(k) =P_{\LCS(F'_{d_n,n}(k))}(X_{n+1})=\sum_{i=1}^{k} \beta_{d_{n+1-i},k,i} X_{d_{n+1-i},n+1-i} . \label{altk}
		\end{align}
		
	 We start with the following Proposition, which is an infinite-dimensional extension to Proposition~2.2 in \cite{mitchell}.
	 
	\begin{proposition}\label{prophelp}
		Under the assumptions of Theorem~\ref{dinfty} the following assertions hold:\\
		(i) The operators $\beta_{d_{n+1-i},n,i}$ from (\ref{alternative}) and $\theta_{d_{n+1-i},n,i}$ from (\ref{xdhat1}) are for $n\in\N$ related by
		\begin{align}
		\theta_{d_{n+1-i},n,i}=\sum_{j=1}^i \beta_{d_{n+1-j},n,j} \theta_{d_{n+1-i},n-j,i-j}, \quad   i=1,\dots,n.\label{link}
		\end{align}
	Furthermore, for every $i, j \in \mathbb{N}$ and $x\in H$, as $n\rightarrow \infty$,
	\begin{compactenum}
		\item[(ii)] $\big\Vert(\beta_{d_{n+1-i},n,i} - \pi_i)( x ) \big\Vert \rightarrow 0$,
		\item[(iii)] $\big\Vert(\beta_{d_{n+1-i},n,i} - \beta_{d_{n+1-i-j},n-j,i})(x) \big\Vert \rightarrow 0$,
		\item[(iv)] $\big\Vert(\theta_{d_{n+1-i},n,i} - \theta_{d_{n+1-i-j},n-j,i} )(x)\big\Vert \rightarrow 0$.
	\end{compactenum}
	\end{proposition}
	
	\begin{proof}
		\textbf{(i)} Set $\theta_{d_{n+1},n,0}:=I_H$. By adding the term $\theta_{d_{n+1},n,0}(X_{d_{n+1},n+1}-\widehat{X}_{d_{n+1},n+1})$ to both sides of \eqref{xdhat1}, we get
		 \begin{align*}
		 X_{d_{n+1},n+1}=\sum_{j=0}^n \theta_{d_{n+1-j},n,j} (X_{d_{n+1-j},n+1-j}-\widehat{X}_{d_{n+1-j},n+1-j}), \quad n\in\N.
		 \end{align*}
		 Plugging this representation of $X_{d_{n+1-i},n+1-i}$ into \eqref{alternative} for $i=1,\dots,n$ yields
		 \begin{align*}
		 \widehat{X}_{d_{n+1},n+1} 
		 &=\sum_{i=1}^n \beta_{d_{n+1-i},n,i} \, \Big(\sum_{j=0}^{n-i} \theta_{d_{n+1-i-j},n-i,j} (X_{d_{n+1-i-j},n+1-i-j}-\widehat{X}_{d_{n+1-i-j},n+1-i-j})\Big).
		 \end{align*}
		 Equating the coefficients of the innovations $(X_{d_{n+1-i},n+1-i} - \widehat{X}_{d_{n+1-i},n+1-i})$ with the innovation representation (\ref{xdhat1}), the identity
		 	\begin{align*}
		 	\wh{X}_{d_{n+1},n+1}= \sum_{i=1}^{n} \theta_{d_{n-i+1},n,i} (X_{d_{n+1-i},n+1-i}-\wh{X}_{d_{n+1-i},n+1-i})
		 	\end{align*}
		 leads by linearity of the operators to \eqref{link}.\\
		 
		\textbf{(ii)} 	Let 
		\begin{align}\label{bpi}
		B_{(d_n),n}=(\beta_{d_{n},n,1},\dots,\beta_{d_1,n,n})\quad\mbox{and}\quad\Pi_n=(\pi_1,\dots,\pi_n),
		\end{align}
		 which are both operators from $H^n$ to $H$ defined as follows: 
		let $x=(x_1,\dots,x_n) \in H^n$ with $x_i\in H$ for $i=1,\dots,n$. Then $B_{(d_{n}),n} \ x =\sum_{i=1}^n\beta_{d_{n+1-i},n,i}  x_i \ \in H$. 
		By definition of the norm in $H^n$ we have for all $x \in H$
		\begin{align*}
		\Vert (B_{(d_{n}),n} - \Pi_n)(x)\Vert  &= \sum_{i=1}^n \Vert(\beta_{d_{n+1-i},n,i} - \pi_i)( x ) \big\Vert.
		\end{align*} 
		We show that this tends to 0 as $n\rightarrow\infty$,
		which immediately gives $\Vert(\beta_{d_{n+1-i},n,i} - \pi_i)( x ) \big\Vert \rightarrow 0$ for all $i\in\N$.  
		
		 First notice that for $x\in H^n$ and with $P_{(d_n)}$ defined in \eqref{pdn}, the triangular inequality yields
		 \begin{align*}
		 \Vert (B_{(d_{n}),n} - \Pi_n)(x)\Vert &\leq  \Vert (B_{(d_{n}),n} - \Pi_n P_{(d_n)})(x)\Vert +  \Vert \Pi_n(I_{H^n}-P_{(d_n)})(x)\Vert\\ & =: J_1(d_n,n)(x)+ J_2(d_n,n)(x),
		 \end{align*}
		 with identity operator $I_{H^n}$ on $H^n$. 
		 We find bounds for $J_1(d_n,n)(x)$ and $J_2(d_n,n)(x)$. 
		 Since uniform convergence implies pointwise convergence, we consider the operator norm of $J_1(d_n,n)(x)$ 
		 \begin{align*}
	 J_1(d_n,n): =\Vert B_{(d_{n}),n} - \Pi_n P_{(d_n)}\Vert_\call 
		 \end{align*}
		 and show that $J_1(d_n,n)\rightarrow 0$ as $n\rightarrow\infty$. 
		 From Theorem~2.1.8 in \cite{simon} we find
		 \begin{align}
		 \Vert  B_{(d_{n}),n} - \Pi_n P_{(d_n)}\Vert_\call^2 = \Vert (B_{(d_{n}),n} - \Pi_nP_{(d_n)})(B_{(d_{n}),n} - \Pi_nP_{(d_n)})^* \Vert_\call.\label{proof5}
		 \end{align}
		Recall the spectral representation of $\Gamma_{(d_n),n}$ as in \eqref{specrep}. By the definition of $B_{(d_{n}),n}$ and $\Pi_n P_{(d_n)}$, note that $(B_{(d_{n}),n} - \Pi_n P_{(d_n)})P_{(d_n)}=B_{(d_{n}),n} - \Pi_n P_{(d_n)}$.
		Extracting the smallest  positive eigenvalue $ \lambda^{(n)}_{(d_n),k_n}$  of $\Gamma_{(d_n),n}$, we get
		 \begin{align}
		 \Big\Vert (B_{(d_{n}),n}& - \Pi_n P_{(d_n)})\Gamma_{(d_n),n}(B_{(d_{n}),n} - \Pi_n P_{(d_n)})^*\Big \Vert_\call \notag \\
		 &  =\Big\Vert (B_{(d_{n}),n} - \Pi_n P_{(d_n)})\sum_{j=1}^{k_n} \lambda_{(d_n),j}^{(n)}(e^{(n)}_{(d_n),j} \otimes e^{(n)}_{(d_n),j}) (B_{(d_{n}),n} - \Pi_n P_{(d_n)})^*\Big \Vert_\call \notag \\
		 & \geq \lambda^{(n)}_{(d_n),k_n}\Big\Vert (B_{(d_{n}),n} - \Pi_n P_{(d_n)})(B_{(d_{n}),n} - \Pi_n P_{(d_n)})^*\Big \Vert_\call.\label{proof14}
		 \end{align} 
		 Since $A_{d_i}\subseteq A_{d_n}$ for all $i\leq n$ we obtain $\mathbf{A}_{(d_n)}:=(A_{d_n},A_{d_{n-1}},\dots,A_{d_1})\subseteq \mathbf{A}_{d_n}:= (A_{d_n},A_{d_{n}},\dots,A_{d_n})$ and, therefore, $P_{d_n}P_ {(d_n)}=P_{(d_n)}$. Together with the definition of $\Gamma_{(d_n),n} $ this implies
		 $$\Gamma_{(d_n),n}  =P_{(d_n)} \Gamma_n P_{(d_n)}=P_{(d_n)}P_{d_n} \Gamma_n P_{d_n}P_{(d_n)}=P_{(d_n)} \Gamma_{d_n,n} P_{(d_n)}.$$
		 Since $\langle x, \Gamma_{(d_n),n} x \rangle = \langle x, \Gamma_{d_n,n} x \rangle$ for all $x\in \mathbf{A}_{d_n}$, and $\mathbf{A}_{(d_n)}\subseteq \mathbf{A}_{d_n}$, we get
		 \begin{align*}
		  \lambda^{(n)}_{{(d_n)},k_n}  =  \min_{x\in \mathbf{A}_{(d_n)}} \frac{\langle x, \Gamma_{(d_n),n} x \rangle}{\Vert x\Vert^2} = \min_{x\in \mathbf{A}_{(d_n)}} \frac{\langle x, \Gamma_{d_n,n} x \rangle}{\Vert x\Vert^2}\geq  \min_{x\in \mathbf{A}_{d_n}} \frac{\langle x, \Gamma_{d_n,n} x \rangle}{\Vert x\Vert^2} =  \lambda^{(n)}_{{d_n}, d_n\cdot n},
		 \end{align*}
		 where the first and last equality hold by application of Theorem~4.2.7 in \cite{hsing}.
	Furthermore, by Lemma~\ref{specdens}, $\lambda^{(n)}_{d_n,d_n\cdot n}\ge 2\pi \alpha_{d_n} $. 
Therefore,
\begin{align}
		\lambda^{(n)}_{{d_n},k_n} \geq\lambda^{(n)}_{d_n,d_n\cdot n}\geq 2\pi \alpha_{d_n}.\label{proof15}
				\end{align}
		With \eqref{proof14} and \eqref{proof15}, we get
		 \begin{align}
	 \Vert  B_{(d_{n}),n} - \Pi_n P_{(d_n)}\Vert^2
		 &\leq  \frac{1}{2\pi\al_{d_n}}\Big\Vert (B_{(d_{n}),n} - \Pi_n P_{(d_n)})\Gamma_{(d_n),n}(B_{(d_n),n} - \Pi_n P_{(d_n)})^*\Big \Vert_\call =: J_1'(d_n,n). \label{proof4}
		 \end{align}
		Furthermore, since $\langle Ax,y\rangle=\langle x, A^*y\rangle$ for $A\in \mathcal{L}$ and $x,y \in H$, and by \eqref{bpi} and the structure of $ \Gamma_{(d_n),n}$,
		 \begin{align*}
	&\Big\Vert (B_{(d_{n}),n} - \Pi_nP_{(d_n)})\Gamma_{(d_n),n}(B_{(d_{n}),n} - \Pi_nP_{(d_n)})^*\Big \Vert_\call\\
		  &\quad \le \Big\Vert
		 \mathbb{E}\Big[
		 \sum_{i=1}^n(\beta_{d_{n+1-i},n,i} - \pi_iP_{A_{d_{n+1-i}}})X_{d_{n+1-i},n+1-i}  \otimes   \sum_{j=1}^n (\beta_{d_{n+1-j},n,j} - \pi_jP_{A_{d_{n+1-j}}})X_{d_{n-j+1},n-j+1}  
		 \Big]\Big\Vert_\call.
		 \end{align*}
		 Now with \eqref{invertible} and \eqref{alternative} we get  
		 \begin{align}
		&\left\Vert
				 \mathbb{E}\Big[
				 \sum_{i=1}^n(\beta_{d_{n+1-i},n,i} - \pi_iP_{A_{d_{n+1-i}}})X_{d_{n+1-i},n+1-i}  \otimes   \sum_{j=1}^n (\beta_{d_{n+1-j},n,j} - \pi_jP_{A_{d_{n+1-j}}})X_{d_{n-j+1},n-j+1}  
				 \Big]\right\Vert_\call \notag\\
		 & =\Big\Vert
		 \mathbb{E}\Big[\Big(\widehat{X}_{d_{n+1},n+1} - X_{n+1} + \varepsilon_{n+1} +  \sum_{i>n}\pi_i X_{n+1-i} + \sum_{i=1}^n\pi_i(I-P_{A_{d_{n+1-i}}})X_{n+1-i}\Big) \notag  \\
		 & \qquad \otimes  \Big(\widehat{X}_{d_{n+1},n+1} - X_{n+1}+ \varepsilon_{n+1} +  \sum_{j>n}\pi_jX_{n+1-j} + \sum_{j=1}^n\pi_j(I-P_{A_{d_{n+1-j}}})X_{n+1-j}\Big)  \Big]\Big\Vert_\call \label{proof9}
		 \end{align}
		 With the trianglular inequality, \eqref{proof9} decomposes in the following four terms giving with \eqref{proof4}: 
		 \begin{align*}
	&	 {2\pi \alpha_{d_n}}   \Vert  B_{(d_{n}),n} - \Pi_n P_{(d_n)}\Vert^2 \\
	& \ \leq \Big\Vert
		 \mathbb{E}\Big[\big(\widehat{X}_{d_{n+1},n+1} - X_{n+1} + \varepsilon_{n+1} \big)\otimes\big( \widehat{X}_{d_{n+1},n+1} - X_{n+1}+ \varepsilon_{n+1} \big)\Big]\Big\Vert_\call \\
		 &\qquad + \Big\Vert
		 \mathbb{E}\Big[\Big(\sum_{i>n}\pi_i X_{n+1-i} + \sum_{i=1}^n\pi_i(I-P_{A_{d_{n+1-i}}})X_{n+1-i}\Big) \otimes\Big(  \sum_{j>n}\pi_jX_{n+1-j} + \sum_{j=1}^n\pi_j(I-P_{A_{d_{n+1-j}}})X_{n+1-j} \Big)\Big]\Big\Vert_\call \\
		 &\qquad +\Big\Vert
		 \mathbb{E}\Big[\Big(\widehat{X}_{d_{n+1},n+1} - X_{n+1} + \varepsilon_{n+1} \Big)\otimes \Big(\sum_{j>n}\pi_jX_{n+1-j} + \sum_{j=1}^n\pi_j(I-P_{A_{d_{n+1-j}}})X_{n+1-j}\Big) \Big]\Big\Vert_\call \\
		 &\qquad + \Big\Vert
		 \mathbb{E}\Big[ \Big(\sum_{i>n}\pi_i X_{n+1-i} + \sum_{i=1}^n\pi_i(I-P_{A_{d_{n+1-i}}})X_{n+1-i} \Big) \otimes\Big( \widehat{X}_{d_{n+1},n+1} - X_{n+1}+\varepsilon_{n+1} \Big) \Big]\Big\Vert_\call.
		 \end{align*}
		 Define 
		 \begin{align*}
		 f(n,d_n,m_n):=\big(\sum_{j>m_n} \Vert \pi_j \Vert_{\call} +	\sum_{j>d_{n-m_n}} \lambda_j \big).
		 \end{align*}
		 By Theorem~\ref{dinfty} the first term is of the order $f(n,d_n,m_n)$. 
		 The second term is of the same order by the calculations following \eqref{proof0}. Concerning the remaining two terms, using first that $\Vert C_{X,Y}\Vert_{\call} \leq \E \Vert X \Vert \Vert Y \Vert $, and then applying the Cauchy-Schwarz inequality gives
		 \begin{align}\label{6.31a}
		& \Big\Vert\mathbb{E}\Big[\Big(\widehat{X}_{d_{n+1},n+1} - X_{n+1} + \varepsilon_{n+1}\Big)\otimes \Big(  \sum_{j>n}\pi_jX_{n+1-j} + \sum_{j=1}^n\pi_j(I-P_{A_{d_{n+1-j}}})X_{n+1-j} \Big) \Big]\Big\Vert_\call^2 \notag\\
		& \qquad \leq \Big( \E \Big\Vert \widehat{X}_{d_{n+1},n+1} - X_{n+1} + \varepsilon_{n+1}\Big\Vert_\call \Big \Vert\sum_{j>n}\pi_jX_{n+1-j} + \sum_{j=1}^n\pi_j(I-P_{A_{d_{n+1-j}}})X_{n+1-j} \Big\Vert \Big)^2\\
		& \qquad \leq  \E \Big\Vert \widehat{X}_{d_{n+1},n+1} - X_{n+1} + \varepsilon_{n+1}\Big\Vert_\call^2 \E \Big \Vert\sum_{j>n}\pi_jX_{n+1-j} + \sum_{j=1}^n\pi_j(I-P_{A_{d_{n+1-j}}})X_{n+1-j} \Big\Vert_\call ^2.\notag
		 \end{align}
		 Both terms are of the order  $f(n,d_n,m_n)$ by Theorem~\ref{dinfty}(i).\\
		 Hence, $\Vert  B_{(d_{n}),n} - \Pi_n P_{(d_n)}\Vert^2$ is of the order  $f(n,d_n,m_n)/\alpha_{d_n}$, and with the assumption \eqref{assumptionspec},  
		 \begin{align}
		J_1(d_n,n)^2  \rightarrow 0,\quad n\to\infty.\label{J1}
		 \end{align}
		 We now estimate $J_2(d_n,n)(x)$, which we have to consider pointwise. For every $x=(x_1,\dots,x_n)\in H^n$ with $x_i\in H$ for $1\leq i\leq n$ and $\Vert x \Vert \leq 1$,
		 \begin{align*}
		 J_2(d_n,n) &= \Vert\Pi_n(I-P_{(d_n)})(x) \Vert \notag\\
		 &=\Big\Vert \big(\pi_1(I_H-P_{A_{d_n}}),\pi_2 (I_H-P_{A_{d_{n-1}}}),\dots, \pi_n (I_H-P_{A_{d_1}})\big) (x) \Big\Vert \\
		 &= \sum_{i=1}^n  \Vert \pi_i(I_H-P_{A_{d_{n+1-i}}})(x_i)  \Vert. 
		 \end{align*}
		 Let $m\in \N$ such that $m<n$. Then,
		 \begin{align}
		 \sum_{i=1}^n  \big\Vert \pi_i(I_H-P_{A_{d_{n+1-i}}})(x_i) \big \Vert &=\sum_{i=1}^{m} \big\Vert \pi_i(I_H-P_{A_{d_{n+1-i}}})(x_i) \big \Vert + \sum_{i=m+1}^{n} \big\Vert \pi_i(I_H-P_{A_{d_{n+1-i}}})(x_i) \big \Vert. \label{proof11} 
		 \end{align}
		 Note that $I_H-P_{A_{d_{n}}}$ is a projection operator on the orthogonal complement of $A_{d_{n}}$. Hence for all $n\in\N$, we have $\Vert I_H-P_{A_{d_n}}\Vert =1$ (see e.g. Theorem~2.1.9 in \cite{simon}). Furthermore, for $A,B\in\call$ and $x\in H$, $\Vert A B x\Vert \leq \Vert A \Vert \Vert B\Vert \Vert x\Vert$, and since $\Vert x_i\Vert \leq 1$,
		 \begin{align}
		 \sum_{i=m+1}^{n} \big\Vert \pi_i(I_H-P_{A_{d_{n+1-i}}})(x_i) \big \Vert
		 &\leq\sum_{i>m}\Vert \pi_i \Vert_\call. \label{proof12}
		 \end{align}
		 Furthermore, since $A_{d_j}\subseteq A_{d_i}$ for $j\leq i$, 
		 \begin{align}
		\sum_{i=1}^{m} \big\Vert \pi_i(I_H-P_{A_{d_n+1-i}})(x_i) \big \Vert &\leq  \sum_{i=1}^{m} \Vert \pi_i \Vert_\call \Vert (I_H-P_{A_{d_n+1-m}})(x_i) \Vert.  \label{proof13} 
		\end{align}
		 Since $\sum_{i=1}^{\infty} \Vert \pi_i\Vert_\call <\infty$, for every $\delta >0$ there exists some $m_{\delta}\in \N$, such that
		 $\sum_{i>m_{\delta}} \Vert \pi_i \Vert_\call<\delta/2.$ 
		 Hence, with \eqref{proof11}, \eqref{proof12} and \eqref{proof13} we estimate
		 \begin{align}
		  \sum_{i=1}^n  \big\Vert \pi_i(I_H-P_{A_{d_{n+1-i}}})(x_i) \big \Vert \leq \sum_{i=1}^{m_{\delta}} \Vert \pi_i \Vert_\call \Vert (I_H-P_{A_{d_{n+1-m_{\delta}}}})(x_i) \Vert + \delta/2 .\label{proof7}
		 \end{align}
		 Furthermore, for the first term of the rhs of \eqref{proof7}, 
		 \begin{align*}
		 \sum_{i=1}^{m_{\delta}} \Vert \pi_i \Vert_\call \Vert (I_H-P_{A_{d_{n+1-m_{\delta}}}})(x_i) \Vert \leq \max_{1\leq j \leq m_{\delta}}  \|( I_H-P_{A_{d_{n+1-m_{\delta}}}}) (x_j)\|  \sum_{i=1}^{m_{\delta}} \Vert \pi_i \Vert_\call .
		 \end{align*} 
		 Now note that  $\Vert (I_H-P_{A_{d_{n+1-m_{\delta}}}})(x) \Vert \rightarrow 0 $ for $n\rightarrow \infty$ for all $x\in H$. 
		 Hence, there exists some $n_{\delta}\in \N$ such that 
		 $$\max_{1\leq j \leq m_{\delta}} \|( I_H-P_{A_{d_{{n_{\delta}}+1-m_{\delta}}}})(x_j) \| \sum_{i=1}^{m_\delta} \Vert \pi_i \Vert_\call<\delta/2.$$
		 Hence, $J_2(d_{n},n)(x)<\delta$ for all $n\ge n_\delta$ and all $x\in H$. \\
		 Together with \eqref{J1}, this proves (ii).\\	
		 \textbf{(iii)}  Similarly to the proof of (ii), we start by defining for every $n\in\N$,
		 	\begin{align*}
		 		\wt B_{(d_n),n-j}:=(\beta_{d_n,n-j,1}, \beta_{d_{n-1},n-j,2},\dots, \beta_{d_{j+1},n-j,n-j},0_H,\dots,0_H), \quad  j=1,\dots,n,
		 	\end{align*}
		 	where the last $j$ entries are $0_H$, the null operator on $H$.
		 	Then $\wt B_{(d_n),n-j}$ is a bounded linear operator from $H^n$ to $H$. 
Analogously to the beginning of the proof of (ii), we show that $\Vert \wt B_{(d_n),n} - \wt B_{(d_n),n-j} \Vert_\call \rightarrow 0$ for $n\rightarrow\infty$. 
With the same calculation as deriving   \eqref{proof4} from \eqref{proof5}, we obtain
		 	\begin{align*}
		 		\Vert \wt B_{(d_n),n} - \wt B_{(d_n),n-j} \Vert_\call ^ 2 
		 	&	\leq \frac{1}{2\pi\alpha_{d_n}} \Vert( \wt B_{(d_n),n} - \wt B_{(d_n),n-j})\Gamma_{(d_n),n}( \wt B_{(d_n),n} - \wt B_{(d_n),n-j})^* \Vert_\call \\
		 	&=: \frac{1}{2\pi\alpha_{d_n}} \wt J'_1(d_n,n).
		 	\end{align*}
		 	Applying the same steps as when bounding $J_1(d_n,n)$ in the proof of (ii), and setting $\beta_{d_{n+j},n,m}=0$ for $m>n$, we obtain
		 	\begin{align*}
		 		\wt J_1'(d_n,n)&=   \Big\Vert
		 		\mathbb{E}\Big[\big(\sum_{i=1}^n(\beta_{d_{n-i+1},n,i} -\beta_{d_{n-i+1},n-j,i})X_{d_{n+1-i},n+1-i} \big)\\
		 		&\qquad \qquad\otimes \big(\sum_{l=1}^n (\beta_{d_{n-l+1},n,l} -\beta_{d_{n-l+1},n-j,l})X_{d_{n+1-l},n+1-l}\big) \Big]\Big\Vert_\call\\
		 		&=   \left\Vert
		 		\mathbb{E}\Big[\big(\widehat{X}_{d_{n+1},n+1}- \widehat{X}_{d_{n+1},n+1}(n-j)\big) \ \otimes (\widehat{X}_{d_{n+1},n+1}- \widehat{X}_{d_{n+1},n+1}(n-j))\Big]\right\Vert_\call,
		 	\end{align*}
		 	where $\widehat{X}_{d_{n+1},n+1}(k)=\sum_{l=1}^k\beta_{d_{n-l+1},k,l}X_{d_{n+1-l},n+1-l}$ is defined as in \eqref{altk}. 
		 	By adding and subtracting $X_{d_{n+1},n+1}+\varepsilon_{n+1}$ and then using the linearity of the scalar product we get
		 	\begin{align*}
		 		\wt J_1'(K,n) =&  \Big\Vert
		 		\mathbb{E}\Big[\big((\widehat{X}_{d_{n+1},n+1}-X_{d_{n+1},n+1}-\varepsilon_{n+1})-( \widehat{X}_{d_{n+1},n+1}(n-j)-X_{d_{n+1},n+1}-\varepsilon_{n+1})\big) \\ 
		 		& \quad \otimes \big((\widehat{X}_{d_{n+1},n+1}-X_{d_{n+1},n+1}-\varepsilon_{n+1})-( \widehat{X}_{d_{n+1},n+1}(n-j)-X_{d_{n+1},n+1}-\varepsilon_{n+1})\big)\Big]\Big\Vert_\call\\
		 		&\leq\   \Big\Vert\mathbb{E}\Big[\big(\widehat{X}_{d_{n+1},n+1}-X_{d_{n+1},n+1}-\varepsilon_{n+1}\big) \otimes \big(\widehat{X}_{d_{n+1},n+1}-X_{d_{n+1},n+1}-\varepsilon_{n+1}\big) \Big]\Big\Vert_\call\\
		 		& \ + \Big\Vert\mathbb{E}\Big[\big( \widehat{X}_{d_{n+1},n+1}(n-j)-X_{d_{n+1},n+1}-\varepsilon_{n+1}\big)  \otimes \big( \widehat{X}_{d_{n+1},n+1}(n-j)-X_{d_{n+1},n+1}-\varepsilon_{n+1})\big] \Big\Vert_\call\\
		 		& \ +\Big\Vert\mathbb{E}\Big[\big(\widehat{X}_{d_{n+1},n+1}-X_{d_{n+1},n+1}-\varepsilon_{n+1}) \otimes \big( \widehat{X}_{d_{n+1},n+1}(n-j)-X_{d_{n+1},n+1}-\varepsilon_{n+1})\Big] \Big\Vert_\call\\
		 		& \ +\Big\Vert\mathbb{E}\Big[\big(\widehat{X}_{d_{n+1},n+1}(n-j)-X_{d_{n+1},n+1}-\varepsilon_{n+1}\big) \otimes \big(\widehat{X}_{d_{n+1},n+1}-X_{d_{n+1},n+1}-\varepsilon_{n+1}\big) \Big]\Big\Vert_\call.
		 	\end{align*}
		 	For $n\rightarrow\infty$ the first term converges to $0$ by Theorem~\ref{dinfty}~(i). 
		 	For every fixed $j\in\{1,\dots,n\}$ the second term converges to $0$ by exactly the same arguments. Similar arguments as in the proof of (ii) show that the third and fourth terms also converge to $0$. 
		 	Indeed,  applying the Cauchy-Schwarz inequality, we find as in \eqref{6.31a},
		 	\begin{align*}
		 		\Big\Vert\mathbb{E}\Big[\big(&\widehat{X}_{d_{n+1},n+1}-X_{d_{n+1},n+1}-\varepsilon_{n+1}\big)  \otimes \big( \widehat{X}_{d_{n+1},n+1}(n-j)-X_{d_{n+1},n+1}-\varepsilon_{n+1}\big) \Big]\Big\Vert^2_\call\\
		 		&\leq  \mathbb{E}\big\Vert\widehat{X}_{d_{n+1},n+1}-X_{d_{n+1},n+1}-\varepsilon_{n+1}\big\Vert_\call^2 \ 
		 		 \mathbb{E}\big\Vert\widehat{X}_{d_{n+1},n+1}(n-j)-X_{d_{n+1},n+1}-\varepsilon_{n+1} \big\Vert_\call^2.
		 	\end{align*}
		 	Since both these terms tend to $0$ for $n\rightarrow \infty$, $\wt J'_1(d_n,n)\rightarrow 0$ for $n\rightarrow\infty$, which finishes the proof of (iii).\\
		 	
		 {\bf  (iv)} 
		 	By \eqref{link}
		 		\begin{align*}
		 				\theta_{d_{n+1-k},n,k}=\sum_{l=1}^k \beta_{d_{n+1-l},n,l} \theta_{d_{n+1-k},n-l,k-l}, \quad   k=1,\dots,n,
		 			\end{align*}
		 	and we get $\theta_{d_n,n,1}=\beta_{d_{n},n,1}$. Hence, for $n\rightarrow\infty$ and fixed $j\in\N$,
		 	\begin{align}
		 		\Vert(\theta_{d_n,n,1} - \theta_{d_n,n-j,1})(x)\Vert=\Vert(\beta_{d_{n},n,1}- \beta_{d_{n},n-j,1})(x) \Vert \rightarrow 0. \label{indstart}
		 	\end{align}
		 	For some fixed $j\in\N$ by a shift of \eqref{link}, we obtain
		 	\begin{align}
		 		\theta_{d_{n+1-k},n-j,k}=\sum_{l=1}^k \beta_{d_{n+1-l},n-j,l} \theta_{d_{n+1-l},n-j-l,k-l}. \label{proof16}
		 	\end{align}
With \eqref{proof16} and then the triangular equality after adding and subtracting $\beta_{d_{n+1-l},n,l}\theta_{d_{n+1-l},n-j-l,k-l}(x)$ for $l=1,\dots,k$,
		 	\begin{align*}
		 		\Big\Vert(	\theta_{d_{n+1-k},n,k}-	\theta_{d_{n+1-k},n-j,k})(x)\Big\Vert&=\Big\Vert\Big(\sum_{l=1}^k \beta_{d_{n+1-l},n,l} \theta_{d_{n+1-k},n-l,k-l}-\beta_{d_{n+1-l},n-j,l} \theta_{d_{n+1-l},n-j-l,k-l}\Big)(x)\Big\Vert\\
		 		&\leq\Big\Vert\sum_{l=1}^k \beta_{d_{n+1-l},n,l}( \theta_{d_{n+1-k},n-l,k-l}-\theta_{d_{n+1-l},n-j-l,k-l})(x)\Big\Vert\\ & \quad +\Big\Vert(\beta_{d_{n+1-l},n,l}-\beta_{d_{n+1-l},n-j,l}) \theta_{d_{n+1-l},n-j-l,k-l}(x)\Big\Vert
		 	\end{align*}
		 	By (iii) $\Vert(\beta_{d_{n+1-l},n,l}-\beta_{d_{n+1-l},n-j,l})(x) \Vert \rightarrow 0$ as $n\rightarrow\infty$. 
		 	Furthermore, if for all $l=1,\dots,i-1$, $\Vert(\theta_{d_{n+1-l},n,l}-\theta_{d_{n+1-l},n-j,l}(x) \Vert \rightarrow 0$, then $\Vert(\theta_{d_{n+1-i},n,i}-\theta_{d_{n+1-i},n-j,i}(x) \Vert \rightarrow 0$. 
		 	The proof then follows by induction with the initial step given in (\ref{indstart}).
		 	\end{proof}
		 	
		 	We are now ready to prove Theorem~\ref{dinfty}(ii).		 	
		 	
		 	\begin{proof}[Proof of Theorem~\ref{dinfty}(ii)]
		 	 Set $\pi_0:=-I_H$. By \eqref{invertible} and the definition of a linear process \eqref{process} 
		 	\begin{align*}
		 	-\varepsilon_n&= \sum_{i=0}^{\infty} \pi_i (X_{n-i})=\sum_{i=0}^{\infty} \pi_i \big(\sum_{j=0}^{\infty} \psi_j \varepsilon_{n-i-j}\big),\quad n\in\Z.
		 	\end{align*}
		    Setting $k=i+j$, this can be rewritten as
		 	\begin{align*}
		 	-\varepsilon_n &=\sum_{i=0}^{\infty} \pi_i \big(\sum_{j=0}^{\infty} \psi_j \varepsilon_{n-i-j}\big) =\sum_{k=0}^{\infty}\big(\sum_{i+j=k} \pi_j \psi_{i} \big)\varepsilon_{n-k} 
		 	=  \sum_{k=0}^{\infty} \sum_{j=0}^k \pi_j\psi_{k-j} \varepsilon_{n-k}.
		 	\end{align*}
		 	Equating the coefficients we get $\sum_{j=0}^{k} \pi_j \psi_{k-j}=0$ for $k>0$. Since $-\pi_0=I_H$, extracting the first term of the series, $\sum_{j=1}^{k} \pi_j \psi_{k-j} -I_H\psi_k=0$, hence,
		 	\begin{align*}
		 	 \sum_{j=1}^{k} \pi_j \psi_{k-j}=\psi_k.
		 	\end{align*}
		 	Furthermore, by \eqref{link} we get for all $x \in H$,
		 	\begin{align*}
		 		\Big\Vert\big(\theta_{d_{n+1-i},n,i}-\psi_{i}\big)(x)\Big\Vert
		 		&=\Big\Vert\Big(\sum_{j=1}^i \beta_{d_{n+1-j},n,j} \theta_{d_{n+1-i},n-j,i-j}- \sum_{j=1}^i \pi_j \psi_{i-j}\Big) (x)\Big\Vert\\
		 		&=\Big\Vert\sum_{j=1}^i (\beta_{d_{n+1-j},n,j}-\pi_j) \theta_{d_{n+1-i},n-j,i-j}(x) - \sum_{j=1}^i \pi_j( \psi_{i-j}-\theta_{d_{n+1-i},n-j,i-j}) (x)\Big\Vert\\
		 		&\leq\Big\Vert\sum_{j=1}^i (\beta_{d_{n+1-j},n,j}-\pi_j) \theta_{d_{n+1-i},n-j,i-j}(x)\Big\Vert + \Big\Vert\sum_{j=1}^i \pi_j( \psi_{i-j}-\theta_{d_{n+1-i},n,i-j}) (x)\Big\Vert\\
		 		& \quad +\Big\Vert\big(\sum_{j=1}^i \pi_j (\theta_{d_{n+1-i},n,i-j}-\theta_{d_{n+1-i},n-j,i-j})\big)(x)\Big\Vert,
		 	\end{align*}
		 	where we have added and subtracted $\theta_{d_{n+1-i},n,i-j}$ and applied the triangular inequality for the last equality.
		 	Now, for $n \rightarrow \infty$, the last term tends to $0$ by Proposition~\ref{prophelp}~(iv). 
		 	The first term tends to $0$  by Proposition~\ref{prophelp}~(ii). 
		 	The second term tends to $0$ by induction, where the initial step is clear, since $\psi_1=-\pi_1$ and $\theta_{d_n,n,1}=\beta_{d_{n},n,1}$.
		 \end{proof}

\bibliographystyle{plain}
\bibliography{bibliography}

\end{document}